\documentclass[11pt]{article}

\usepackage[margin=1in]{geometry}
\usepackage{amsmath,amssymb,amsthm,mathtools}
\usepackage{bbm}
\usepackage{enumitem, comment}
\usepackage{hyperref}
\usepackage[nameinlink,capitalize,noabbrev]{cleveref}
\usepackage{graphicx}

\usepackage{xcolor}

\theoremstyle{plain}
\newtheorem{theorem}{Theorem}[section]
\newtheorem{lemma}[theorem]{Lemma}
\newtheorem{proposition}[theorem]{Proposition}
\newtheorem{corollary}[theorem]{Corollary}

\theoremstyle{definition}
\newtheorem{definition}[theorem]{Definition}

\theoremstyle{remark}
\newtheorem{remark}[theorem]{Remark}

\newcommand{\E}{\mathbb{E}}
\renewcommand{\P}{\mathbb{P}}

\newcommand{\1}{\mathbbm{1}}
\newcommand{\KL}{\mathrm{D}}

\newcommand{\Cov}{\mathrm{Cov}}
\newcommand{\Dir}{\mathrm{Dir}}
\newcommand{\Beta}{\mathrm{Beta}}
\newcommand{\Bin}{\mathrm{Bin}}
\newcommand{\Gap}{\mathrm{Gap}}
\newcommand{\Gam}{\mathrm{Gamma}}
\newcommand{\supp}{\mathrm{supp}}

\newcommand{\diag}{\ \mathrm{diag}}

\newcommand{\Perrw}{\P_{\mathrm{errw}}}
\newcommand{\Psrw}{\P_{\mathrm{srw}}}
\newcommand{\Eerrw}{\E_{\mathrm{errw}}}
\newcommand{\Esrw}{\E_{\mathrm{srw}}}

\newcommand{\Hent}{\mathrm{H}}

\newif\ifva
\vatrue

\title{
An Information-theoretic Analysis of Edge-reinforced \\ Random Walks
}
\author{Qinghua (Devon) Ding, \quad Venkat Anantharam\\
  Department of Electrical Engineering and Computer Sciences\\
                    University of California at Berkeley\\
                    Berkeley, CA, United States\\
                    Email: \{devon\_ding, ananth\}@berkeley.edu}

\begin{document}
\maketitle

\begin{abstract}
Reinforced random walks are 
random walks on graphs
whose transition probabilities along edges from a vertex
are proportional to the weights of those edges,
but where the weight of an edge evolves
in a way that depends on the past traversals across it. 
In an edge-reinforced random walk (ERRW) the weight of an edge increases by $1$ 
whenever that edge is traversed, in either direction.
On a finite graph, an ERRW admits a remarkable representation as a random walk in a random environment. The law of the environment
is given by the so-called {\em magic formula}~\cite{KeaneRolles2000,MerklOryRolles2008},
with this law depending on the initial edge weights.
This representation provides a natural route for studying 
statistical
properties of ERRWs~\cite{DiaconisRolles2006,SabotTarres2015,SabotTarresZeng2017}.

This work focuses on 
various
information-theoretic 
quantities associated with ERRWs on finite graphs, motivated in part by the problem of statistically distinguishing between different ERRW models from observed trajectories. 
In particular, we study the entropy rate of 
an ERRW.
We also study
the Kullback--Leibler divergence (KL divergence) between two ERRW environment laws, and the KL divergence between the corresponding finite-trajectory distributions. Leveraging structural properties of the underlying random environment~\cite{SabotTarres2015,SabotTarresZeng2017}, we derive an annealed representation of the entropy rate, a closed-form formula for the environment-level KL divergence, and quantitative bounds on the convergence of trajectory-level KL divergence toward environment-level KL divergence.

These information-theoretic quantities are motivated by the two-point hypothesis testing problem for ERRW trajectories, and in particular by the associated Stein exponent. We also expect them to play a fundamental role in the study of other testing problems for ERRWs, including identity testing and closeness testing.
\end{abstract}

\section{Introduction}
\label{sec:intro}


An edge-reinforced random walk (ERRW) is a
random walk on a graph
whose transition probabilities along an edge from a vertex
are proportional to the weights of those edges, and where
the weight of an edge increases by $1$ 
whenever that edge is traversed, in either direction.
Introduced in the work of Coppersmith and Diaconis, the model has since been studied from several points of view, including 
its intimate connections with the concept of
partial exchangeability; Bayesian statistics for reversible Markov chains; and 
as an example of a random walk in a random environment;
see, for example, \cite{KeaneRolles2000,Rolles2003,DiaconisRolles2006,MerklRolles2006,MerklOryRolles2008}. On finite graphs, a basic structural fact is that ERRW has the same law as a mixture of reversible Markov chains, or equivalently as a random walk in a random environment, with 
an
explicit mixing measure given by the 
so-called
{\em magic formula} \cite{KeaneRolles2000,MerklOryRolles2008}. This representation places ERRW in the broader framework of de Finetti-type results for Markov chains and reversible processes \cite{DiaconisFreedman1980,DiaconisRolles2006}.
This random-environment point of view has also led to the further developments relating ERRW to the 
vertex-reinforced jump process (VRJP) and to random Schr\"odinger operators and the supersymmetric hyperbolic sigma model \cite{SabotTarres2015,SabotTarresZeng2017}. We refer the reader to the survey article of Merkl and Rolles \cite{MerklRolles2006} for background on linearly edge-reinforced random walk and its principal structural properties.

The purpose of the present paper is to study several information-theoretic quantities associated with ERRW on finite graphs. The main objects of interest are the entropy rate of the reinforced walk, the Kullback--Leibler divergence (KL divergence) between two ERRW environment laws, and the KL divergence between the corresponding finite-trajectory laws. These quantities arise naturally when one considers statistical questions for ERRW trajectories. In particular, the trajectory-level KL divergence governs the asymptotic exponent in two-point hypothesis testing, through Stein's lemma. More broadly, the same quantities should be relevant for other testing problems for ERRW, such as identity testing and closeness testing.

Our main results are as follows.
\begin{itemize}
    \item \emph{Entropy rate of ERRW.} We prove that the entropy rate of ERRW admits an annealed representation in terms of the mixing measure.
    We derive an explicit formula for the entropy rate. We also provide an upper bound for it, expressed via the normalizing constant in the magic formula.

    \item \emph{KL divergence between environment laws.} 
    We derive a closed-form formula for the environment-level KL divergence, i.e. the divergence between the mixing measures 
    of ERRW when started from two different initial edge weights. This is based on the fact that the family of ERRW mixing measures forms a canonical exponential family. 
    We show
    that this formula admits a natural probabilistic interpretation through the 
    Sabot-Tarr\`{e}s-Zheng field representation (STZ-field representation)~\cite{SabotTarresZeng2017}: 
    the underlying edge Gamma field factorizes into an ERRW environment component and independent Gamma fields on the vertices, so that the KL identity appears as the net contribution of the edge-Gamma and vertex-Gamma terms.
    For details see \cref{sec:kl_env}.

    \item \emph{Convergence of trajectory-level KL divergences.} We prove that the trajectory-level KL divergence converges to the environment-level KL divergence as \(T\to\infty\) at a sublinear convergence rate on general graphs, with a tight characterization in the urn (\(n\)-star) case. The main technical difficulty is the derivation of quantitative bounds on the random environment. We address this by combining a combinatorial lemma with analytic control of the normalized environment over the simplex, and concentration estimates for finite reversible Markov chains.
\end{itemize}

The results of the paper support the following general point of view. For ERRW on finite graphs, the random environment, especially at the level of the STZ-field~\cite{SabotTarresZeng2017}, is the natural object in terms of which entropy and KL divergence are most naturally understood. In this sense, the random-environment representation is not only a structural description of the model, but also the appropriate framework for its information-theoretic analysis.

The paper is organized as follows. In \Cref{sec:related-work}, we review the most relevant earlier work. In \Cref{sec:prelims}, we recall the definition of an ERRW and its random-environment representation. In \Cref{sec:entropy}, we study the entropy rate of ERRW. In \Cref{sec:kl_env}, we analyze the KL divergence between ERRW environment laws. In \Cref{sec:kl_path}, we study the KL divergence between finite trajectory laws and its relation to posterior updating. We then conclude with a discussion of several open problems.

\section{Related Work}
\label{sec:related-work}

Linearly edge-reinforced random walk was introduced by Coppersmith and Diaconis \cite{CoppersmithDiaconis1986}
and has been studied extensively on both finite and infinite graphs. On finite graphs, a basic structural result is that ERRW has the same law as a mixture of reversible Markov chains, or equivalently as a random walk in a random environment; this representation was established by Keane and Rolles
\cite{KeaneRolles2000},
and the corresponding mixing measure was later identified explicitly through the magic formula by Merkl, \"Ory, and Rolles \cite{KeaneRolles2000,MerklOryRolles2008}. The survey article \cite{MerklRolles2006} provides further background on linearly edge-reinforced random walk and its principal structural properties.

A second line of work concerns the statistical interpretation of ERRW. Rolles 
\cite{Rolles2003}
showed that ERRW arises naturally in the Bayesian analysis of reversible Markov chains, and Diaconis and Rolles developed the corresponding conjugate-prior framework for reversible transition matrices \cite{Rolles2003,DiaconisRolles2006}. In particular, the posterior update of the environment law after observing a path plays an important role in the present paper.

ERRW is also closely related to other reinforced processes and to the random-field representations developed in the work of Sabot, Tarr\`es, and Zeng
\cite{SabotTarres2015,SabotTarresZeng2017}.
In particular, the connections with VRJP, random Schr\"odinger operators, and the supersymmetric hyperbolic sigma model provide a useful structural framework for interpreting the random environment and its factorization properties \cite{SabotTarres2015,SabotTarresZeng2017}. These ideas enter our treatment of the KL divergence between environment laws.

The quantitative part of the paper also draws on earlier estimates for the ERRW environment and on concentration inequalities for finite Markov chains. For the former, we use the simplex-density representation and related bounds due to Merkl and Rolles \cite{MerklRolles2008Bounding}. For the latter, we use standard concentration tools for finite reversible Markov chains, especially the Chernoff-type bound of L\'ezaud, the concentration inequalities of Paulin, and the general background in Levin, Peres, and Wilmer \cite{Lezaud1998,Paulin2015,LevinPeresWilmer2017}.

To the best of our knowledge, however, the entropy rate of ERRW, the closed-form KL divergence between ERRW environment laws, and the quantitative comparison between environment-level and trajectory-level KL divergences have not been developed systematically in the previous literature. The present paper is intended as a contribution in this direction.


\section{Preliminaries on ERRW}
\label{sec:prelims}

We write $:=$ and $=:$ for equality by definition. The notation $\{u,v\}$ denotes an unordered pair of elements.
Thus, the edge between vertices $u$ and $v$ in an undirected graph is denoted $\{u,v\}$.

Throughout the paper, let \(G=(V,E)\) be a finite, connected, simple undirected graph, with \(|V|=:n\) and \(|E|=:m\). For \(u,v\in V\), we write \(u\sim v\) if \(\{u,v\}\in E\), and for \(v\in V\), we write \(\deg(v):=|\{u\in V:\ u\sim v\}|\). We fix a distinguished initial vertex \(v_0\in V\), and we equip the edges with positive initial weights \(A=(a_e)_{e\in E}\in(0,\infty)^E\). For \(v\in V\), let \(a_v:=\sum_{e\ni v} a_e\). When \(e=\{u,v\}\), we occasionally write \(a_e=a_{uv}=a_{vu}\).

\subsection{Definition and basic path statistics}

A trajectory of length \(T\) is denoted by \(X_0^T:=(X_0,\dots,X_T)\), 
where \(X_0=v_0\) and \(X_t\in V\) for $0 \le t \le T$.
For a realized trajectory \(x_0^T\), define the directed transition counts
\begin{equation}        \label{eq:edgecounts}
N_{uv}(x_0^T):=\bigl|\{t\in\{1,\dots,T\}:\ x_{t-1}=u,\ x_t=v\}\bigr|,
\qquad u\sim v,
\end{equation}
the vertex departure counts 
\begin{equation}        \label{eq:departurecounts}
N_u(x_0^T):=\sum_{v\sim u} N_{uv}(x_0^T), 
\end{equation}
and the undirected edge counts 
\begin{equation}    \label{eq:undirectedcounts}
N_{\{u,v\}}(x_0^T):=N_{uv}(x_0^T)+N_{vu}(x_0^T). 
\end{equation}
When no confusion 
arises,
we abbreviate 
the undirected edge counts by \(N_e:=N_e(x_0^T)\).

For each edge \(e\in E\) and time \(t\ge 0\), the edge local time is defined as
\begin{equation}    \label{eq:edgelocaltime}
L_e(t):=a_e+\sum_{s=1}^t \1\bigl(\{X_{s-1},X_s\}=e\bigr)
      =a_e+N_e(X_0^t).
\end{equation}
Here \((a_e, e \in E)\) are the initial edge weights, which are assumed to be strictly positive.
We think of $L_e(t)$ as the edge weight of edge $e$ at time $t$.

\begin{definition}[Linearly edge-reinforced random walk]
Let \(G=(V,E)\), \(v_0\in V\), and \(A=(a_e)_{e\in E}\in(0,\infty)^E\). A process \((X_t)_{t\ge 0}\) is called an \emph{edge-reinforced random walk} (ERRW) on \((G,v_0,A)\) if \(X_0=v_0\) and, for every \(t\ge 0\) and every neighbor \(v\sim X_t\),
\[
\P\bigl(X_{t+1}=v \,\big|\, X_0,\dots,X_t\bigr)
=
\frac{L_{\{X_t,v\}}(t)}{\sum_{u\sim X_t} L_{\{X_t,u\}}(t)}.
\]
\end{definition}

Thus, at each step, the walk chooses among the edges incident to its current position with probabilities proportional to their current weights.

\subsection{ERRW as a random walk in a random environment}
\label{subsec:magic}

A basic structural fact is that ERRW on a finite graph is a mixture of reversible Markov chains; see \cite{KeaneRolles2000,MerklOryRolles2008}. Fix once and for all a reference edge \(e_0\in E\). An environment is represented by a family of positive edge conductances \(w=(w_e)_{e\in E}\in(0,\infty)^E\) in the pinned gauge \(w_{e_0}=1\). We write \(w_{-e_0}:=(w_e)_{e\in E\setminus\{e_0\}}\), \(dw_{-e_0}:=\prod_{e\neq e_0}dw_e\), and \(w_v:=\sum_{e\ni v} w_e\) for \(v\in V\). Let \(\mathcal T\) denote the set of spanning trees of \(G\), and define \(\tau(w):=\sum_{T\in\mathcal T}\prod_{e\in T} w_e\).

\begin{definition}[Magic-formula mixing measure \cite{MerklOryRolles2008}]
\label{def:magic_measure}
Fix \((G,v_0,A)\) and work in the gauge \(w_{e_0}=1\). The \emph{mixing measure} \(\mu_{v_0,A}\) is the probability measure on \(\{w\in(0,\infty)^E:\ w_{e_0}=1\}\) with density
\begin{equation}
\label{eq:mixing_measure_magic}
d\mu_{v_0,A}(w_{-e_0})
=
Z_{v_0,A}^{-1}\,
\frac{w_{v_0}^{1/2}\prod_{e\in E} w_e^{a_e-1}}
{\prod_{v\in V} w_v^{(a_v+1)/2}}
\sqrt{\tau(w)}\,dw_{-e_0},
\end{equation}
where
\begin{equation}
\label{eq:Zv0A}
Z_{v_0,A}
:=
\frac{\prod_{e\in E}\Gamma(a_e)}
{\prod_{v\in V}\Gamma\left(\frac{a_v+1-\1(v=v_0)}{2}\right)}
\cdot
\frac{\pi^{(|V|-1)/2}}{2^{\,1-|V|+\sum_{e\in E} a_e}}.
\end{equation}
\end{definition}
Here $\Gamma(x) := \int_0^\infty t^{x-1} e^{-t} dt$ denotes the value of the 
Gamma function at $x > 0$. Also note that, since $w_{e_0} =1$ we sometimes write
$d\mu_{v_0,A}(w)$ for $d\mu_{v_0,A}(w_{-e_0})$, with the understanding that
an additional factor of $\delta(w_{e_0} -1)$ has been tagged on.

\begin{theorem}[Random-environment representation of ERRW \cite{KeaneRolles2000,MerklOryRolles2008}]
\label{thm:errw_rwre}
Let \((X_t)_{t\ge 0}\) 
denote the ERRW on
\((G,v_0,A)\),
i.e. on the graph $G$ with the initial condition $v_0$ and the initial edge weights given by $A$.
Then there exists a random environment \(W\sim\mu_{v_0,A}\) such that, conditional on \(W=w\), the process \((X_t)_{t\ge 0}\) is the reversible Markov chain on \(V\) with transition kernel
\begin{equation}
\label{eq:Pw_def}
p_{ij}(w)=\P(X_{t+1}=j\mid X_t=i,\ W=w)=\frac{w_{\{i,j\}}}{w_i},
\qquad i\sim j,
\end{equation}
and \(p_{ij}(w)=0\) if \(i\not\sim j\). Equivalently,
\[
\Perrw^{v_0,A}(\cdot)=\int \Psrw^{v_0,w}(\cdot)\,\mu_{v_0,A}(dw),
\]
where \(\Psrw^{v_0,w}\) denotes the law of the Markov chain 
on the vertices of $G$
started from \(v_0\) with kernel \eqref{eq:Pw_def},
and \(\Perrw^{v_0,A}\) denotes the law of the ERRW on \((G,v_0,A)\).
\end{theorem}

For a fixed environment \(w\), the chain \(P_w=(p_{ij}(w))\) is reversible with stationary distribution
\[
\pi_i(w):=\frac{w_i}{\sum_{v\in V} w_v},\qquad i\in V.
\]
Indeed, for \(i\sim j\), one has \(\pi_i(w)p_{ij}(w)=w_{ij}/\sum_{v\in V}w_v=\pi_j(w)p_{ji}(w)\).

The representation in \Cref{thm:errw_rwre} will be the 
basis for our analysis of the entropy rate and KL divergence in the rest of the paper.

\subsection{Quenched likelihood and count representation}

Let \(w\) be a fixed environment, and let \(x_0^T\) be a trajectory with \(x_0=v_0\). Under the quenched law \(\Psrw^{v_0,w}\), the probability of \(x_0^T\) 
for $T \ge 1$
is
\begin{equation}
\label{eq:quenched_likelihood_counts}
\Psrw^{v_0,w}(X_0^T=x_0^T)
=
\prod_{t=1}^T p_{x_{t-1},x_t}(w)
=
\prod_{u\in V}\prod_{v\sim u}
\left(\frac{w_{\{u,v\}}}{w_u}\right)^{N_{uv}(x_0^T)}.
\end{equation}
Equivalently, the quenched log-likelihood 
for $T \ge 1$
is
\begin{equation}
\label{eq:quenched_loglik}
\log \Psrw^{v_0,w}(X_0^T=x_0^T)
=
\sum_{u\in V}\sum_{v\sim u} N_{uv}(x_0^T)\log w_{\{u,v\}}
-
\sum_{u\in V} N_u(x_0^T)\log w_u.
\end{equation}
In particular, for fixed \(w\), the quenched law of the trajectory depends on \(x_0^T\) only through its directed transition count vector \(N(x_0^T):=(N_{uv}(x_0^T))_{u\sim v}\). We shall use this observation repeatedly.

For later use, note also that the stationary mass carried by an undirected edge \(e=\{i,j\}\) 
in the fixed environment \(w\) 
is \(2\pi_i(w)p_{ij}(w)=w_e/\sum_{g\in E}w_g\). Accordingly, we consider the 
stationary occupation measure on edges under the fixed environment \(w\) and denote it
by \(X_e(w):=w_e/\sum_{g\in E} w_g\), \(e\in E\). This quantity will play an important role in the 
rest
of the paper.

\section{Entropy Rate of ERRW}
\label{sec:entropy}

In this section, we study the entropy rate of ERRW through its random-environment representation. Throughout, let \((X_t)_{t\ge 0}\) 
denote the
ERRW on \((G,v_0,A)\), and let \(W\sim\mu_{v_0,A}\) denote the random environment from \Cref{thm:errw_rwre},
where \(\mu_{v_0,A}\) is defined in \eqref{eq:mixing_measure_magic}.
Conditional on \(W=w\), let \(P_w=(p_{ij}(w))\) be the corresponding quenched transition kernel and let \(\pi(w)\) be its stationary distribution.

\subsection{Entropy rate and the finite-state Markov-chain formula}

We write \(\Hent(\cdot)\) for Shannon entropy with 
logarithms to the natural base.
Since \(X_0=v_0\) is deterministic, \(\Hent(X_0^T)\)
for $T \ge 1$
is the same as the entropy of \((X_1,\dots,X_T)\); nevertheless, we continue to write \(\Hent(X_0^T)\) for convenience.

For each \(T\ge 1\) and each environment \(w\), define the quenched path entropy by
\(h_T(w):=\Hent_{\Psrw^{v_0,w}}(X_0^T)\). Since \(V\) is finite, this is well-defined and satisfies \(0\le h_T(w)\le T\log|V|\). We define the conditional entropy of \(X_0^T\) given \(W\) by \(\Hent(X_0^T\mid W):=\E[h_T(W)]\), and the mutual information between \(W\) and \(X_0^T\) by
\(I(W;X_0^T):=\Hent(X_0^T)-\Hent(X_0^T\mid W)\). Equivalently,
\(I(W;X_0^T)=\KL(P_{W,X_0^T}\,\|\,P_W\otimes P_{X_0^T})\), where \(\KL\) denotes relative entropy. Since \(X_0^T\) takes values in the finite set \(V^{T+1}\), all these quantities are finite.

\begin{definition}[Entropy rate]
\label{def:entropy_rate}
Let \((Y_t)_{t\ge 0}\) be a discrete-time process on a finite state space. Its \emph{entropy rate} is
defined as
\[
r:=\lim_{T\to\infty}\frac{1}{T}\Hent(Y_0^T),
\]
provided the limit exists.
\end{definition}

The entropy rate of a finite irreducible Markov chain is well known (e.g., Cover and Thomas~\cite[Sec.~4.2]{cover_thomas_2006}).

\begin{proposition}[Entropy rate of a finite Markov chain]
\label{prop:mc_entropy_rate}
Let \(P=(p_{ij})_{i,j\in V}\) be an irreducible Markov chain on a finite state space \(V\), and let \(\pi\) be its stationary distribution. Then the entropy rate exists and is given by
\[
r(P) :=-\sum_{i\in V}\pi_i\sum_{j\in V} p_{ij}\log p_{ij}.
\]
\end{proposition}

In the present setting, for each fixed environment \(w\), the quenched chain \(P_w\) is irreducible on the finite connected graph \(G\), so Proposition \ref{prop:mc_entropy_rate} applies.

\subsection{Annealed reduction through the random environment}
\label{subsec:anneal}

We first 
bound
the entropy 
\(\Hent(X_0^T)\) under the law of the
ERRW 
on \((G,v_0,A)\) in terms of
an annealed average of the 
corresponding quantities for
the quenched chains
corresponding to each environment \(w\).

\begin{lemma}[Annealed reduction]
\label{lem:entropy_anneal}
For every \(T\ge 1\),
\begin{equation}
\label{eq:anneal_mi}
\Hent(X_0^T)=\Hent(X_0^T\mid W)+I(W;X_0^T).
\end{equation}
Moreover,
\begin{equation}
\label{eq:mi_log_bound}
I(W;X_0^T)\le (2|E|-1)\log(T+1).
\end{equation}
In particular, \(T^{-1}I(W;X_0^T)\to 0\) as \(T\to\infty\).
\end{lemma}

\begin{proof}
The identity \eqref{eq:anneal_mi} is immediate from the definition of mutual information.

To prove \eqref{eq:mi_log_bound}, let \(N(X_0^T):=(N_{uv}(X_0^T))_{u\sim v}\) be the directed transition-count vector. By \eqref{eq:quenched_likelihood_counts}, for every fixed environment \(w\), the quenched probability of a trajectory \(x_0^T\) depends on \(x_0^T\) only through \(N(x_0^T)\). Thus, if \(\mathcal S(n):=\{x_0^T\in V^{T+1}:N(x_0^T)=n,\ x_0=v_0\}\), then \(\Psrw^{v_0,w}(X_0^T=x_0^T)\) is constant on \(\mathcal S(n)\). Hence, conditional on \(\{N(X_0^T)=n\}\), the quenched law of \(X_0^T\) is uniform on \(\mathcal S(n)\), and in particular does not depend on \(w\). Equivalently, \(W\perp X_0^T\mid N(X_0^T)\).

By the chain rule for mutual information,
\(I(W;X_0^T)=I(W;N(X_0^T))+I(W;X_0^T\mid N(X_0^T))=I(W;N(X_0^T))\).
Since \(N(X_0^T)\) is discrete, \(I(W;N(X_0^T))\le \Hent(N(X_0^T))\).

Now \(N(X_0^T)\) has \(2|E|\) nonnegative integer coordinates and total mass \(T\), so its support is contained in the set of weak compositions of \(T\) into \(2|E|\) parts. Therefore
\[
|\supp(N(X_0^T))|\le \binom{T+2|E|-1}{2|E|-1},
\]
and hence
\[
\Hent(N(X_0^T))
\le
\log \binom{T+2|E|-1}{2|E|-1}
\le
(2|E|-1)\log(T+1).
\]
This proves \eqref{eq:mi_log_bound}.
\end{proof}

In the special case of an $n$-star the bound of \eqref{eq:mi_log_bound} is tight up to a constant, see Appendix~\ref{sec:mi_tightness_star}.

We now deduce the existence of the entropy rate of the ERRW on \((G,v_0,A)\)
and provide an annealed representation for this entropy rate.

\begin{theorem}[Annealed entropy-rate formula]
\label{thm:entropy_rate_annealed}
The entropy rate of the ERRW on \((G,v_0,A)\) exists. More precisely, for \(\mu_{v_0,A}\)-a.e.\ environment \(w\), the quenched chain \(P_w\) has entropy rate
\[
r(P_w)=-\sum_{i\in V}\pi_i(w)\sum_{j\in V} p_{ij}(w)\log p_{ij}(w),
\]
and
\begin{equation}
\label{eq:errw_entropy_rate_annealed}
r(v_0,A):=\lim_{T\to\infty}\frac{1}{T}\Hent(X_0^T)=\E\left[r(P_W)\right].
\end{equation}
\end{theorem}

\begin{proof}
Fix an environment \(w\). For \(i\in V\), define the one-step entropy by
\(h_i(w):=-\sum_{j\in V} p_{ij}(w)\log p_{ij}(w)
=-\sum_{j\sim i}\frac{w_{ij}}{w_i}\log\frac{w_{ij}}{w_i}\).
Since under \(\Psrw^{v_0,w}\) the process is Markov, the chain rule for entropy gives
\(h_T(w)=\sum_{t=1}^T \Hent_{\Psrw^{v_0,w}}(X_t\mid X_0^{t-1})\), and the Markov property yields
\(\Hent_{\Psrw^{v_0,w}}(X_t\mid X_0^{t-1})
=\Hent_{\Psrw^{v_0,w}}(X_t\mid X_{t-1})
=\E^{v_0,w}[h_{X_{t-1}}(w)]\).
Therefore
\[
h_T(w)=\sum_{t=0}^{T-1}\E^{v_0,w}[h_{X_t}(w)].
\]

Since \(P_w\) is irreducible on a finite state space, its Ces\`aro averages converge to the stationary distribution \(\pi(w)\). Thus, for each \(i\in V\),
\[
\frac1T\sum_{t=0}^{T-1}\Psrw^{v_0,w}(X_t=i)\longrightarrow \pi_i(w).
\]
Multiplying by \(h_i(w)\) and summing over \(i\), we obtain
\[
\frac{1}{T}h_T(w)
=
\sum_{i\in V}
\left(
\frac1T\sum_{t=0}^{T-1}\Psrw^{v_0,w}(X_t=i)
\right)
h_i(w)
\longrightarrow
\sum_{i\in V}\pi_i(w)h_i(w)=r(P_w).
\]
Moreover, \(0\le h_i(w)\le \log\deg(i)\le \log|V|\), so \(0\le T^{-1}h_T(w)\le \log|V|\) and \(0\le r(P_w)\le \log|V|\).

Now divide \eqref{eq:anneal_mi} by \(T\). By Lemma \ref{lem:entropy_anneal}, the mutual-information term is \(O((\log T)/T)=o(1)\). On the other hand,
\[
\frac1T\Hent(X_0^T\mid W)=\E\left[\frac1T h_T(W)\right].
\]
Since \(T^{-1}h_T(W)\to r(P_W)\) almost surely and \(0\le T^{-1}h_T(W)\le \log|V|\), dominated convergence gives
\[
\frac1T\Hent(X_0^T\mid W)\longrightarrow \E[r(P_W)].
\]
Combining these two limits with \eqref{eq:anneal_mi} proves \eqref{eq:errw_entropy_rate_annealed}.
\end{proof}

\subsection{Exact formulas and bounds for the ERRW entropy rate}

We now rewrite the quenched entropy rate in a form better adapted to the random environment.

\begin{theorem}[Exact formula and upper bound]
\label{thm:annealed_entropy_rate}
Let \(r(v_0,A)\) be the ERRW entropy rate from \Cref{thm:entropy_rate_annealed}. Then
\begin{equation}
\label{eq:annealed_entropy_rate_integral}
r(v_0,A)
=
-\int
\left(
\sum_{\{i,j\}\in E}
\frac{w_{ij}}{\sum_{\{u,v\}\in E} w_{uv}}
\log\frac{w_{ij}}{\sqrt{w_iw_j}}
\right)
\mu_{v_0,A}(dw),
\end{equation}
where \(\mu_{v_0,A}\) is defined in \eqref{eq:mixing_measure_magic}.
For \(v\in V\), write \(b_v:=\frac{a_v+1-\1(v=v_0)}{2}\). Then
\begin{equation}
\label{eq:r_upper_final_clean}
r(v_0,A)
\le
\sum_{e=\{i,j\}\in E}
\bigg(
\frac{a_e}{2}\cdot
\frac{\Gamma(b_i)}{\Gamma(b_i+\tfrac12)}\cdot
\frac{\Gamma(b_j)}{\Gamma(b_j+\tfrac12)}
\bigg)
\bigg(
\log 2+\frac12\Psi\bigl(b_i+\tfrac12\bigr)+\frac12\Psi\bigl(b_j+\tfrac12\bigr)-\Psi(a_e+1)
\bigg).
\end{equation}
Here, for a real number $x > 0$, we use the standard notation $\Gamma(x) := \int_0^\infty t^{x-1} e^{-t} dt$ for the Gamma function
and $\Psi(x) := \frac{d}{dx} \log \Gamma(x)$ for the digamma function.
\end{theorem}

\begin{proof}
By \Cref{thm:entropy_rate_annealed}, \(r(v_0,A)=\E[r(P_W)]\). For a fixed environment \(w\),
\[
r(P_w)
=
-\sum_{i\in V}\pi_i(w)\sum_{j\in V} p_{ij}(w)\log p_{ij}(w)
=
-\frac{1}{\sum_{v\in V}w_v}
\sum_{i\in V}\sum_{j\sim i} w_{ij}\log\frac{w_{ij}}{w_i}.
\]
For a fixed undirected edge \(\{i,j\}\), the two oriented contributions are
\(w_{ij}\log\frac{w_{ij}}{w_i}\) and \(w_{ij}\log\frac{w_{ij}}{w_j}\), whose sum is \(2w_{ij}\log\frac{w_{ij}}{\sqrt{w_iw_j}}\). Summing over all edges and using \(\sum_{v\in V}w_v=2\sum_{e\in E}w_e\), we obtain
\begin{equation}    \label{eq:entropy_conditioned}
r(P_w)
=
-\sum_{\{i,j\}\in E}
\frac{w_{ij}}{\sum_{e\in E} w_e}
\log\frac{w_{ij}}{\sqrt{w_iw_j}}.
\end{equation}
Averaging over \(W\sim\mu_{v_0,A}\) yields \eqref{eq:annealed_entropy_rate_integral}.

For the upper bound, define \(Y_e(w):=\frac{w_e}{\sqrt{w_iw_j}}\) for \(e=\{i,j\}\in E\). Since \(w_e\le w_i\) and \(w_e\le w_j\), we have \(0<Y_e(w)\le 1\). Also, \(\sum_{g\in E} w_g\ge \frac12(w_i+w_j)\ge \sqrt{w_iw_j}\), so
\[
\frac{w_e}{\sum_{g\in E} w_g}\le Y_e(w).
\]
Because \(-\log Y_e(w)\ge 0\), 
\eqref{eq:annealed_entropy_rate_integral} 
and \eqref{eq:entropy_conditioned} 
imply
\[
r(P_w)\le -\sum_{e\in E} Y_e(w)\log Y_e(w),
\qquad
r(v_0,A)\le -\sum_{e\in E}\E_{\mu_{v_0,A}}[Y_e(W)\log Y_e(W)].
\]

Fix \(e=\{i,j\}\in E\), and let \(A^{(e)}:=A+\1_e\). 
Comparing the densities
\eqref{eq:mixing_measure_magic} for \(A\) and \(A^{(e)}\), we see that increasing
\(a_e\) by \(1\) multiplies the unnormalized density by
\[
\frac{w_e}{\sqrt{w_iw_j}}=Y_e(w),
\]
while the remaining change is only the ratio of normalizing constants. This remains
true when \(e=e_0\): in the pinned gauge \(w_{e_0}=1\), the edge factor \(w_{e_0}\)
is equal to \(1\), but the two incident vertex factors \(w_i^{-1/2}\) and
\(w_j^{-1/2}\) are still introduced. Hence the multiplier is still
\[
Y_{e_0}(w)=\frac{w_{e_0}}{\sqrt{w_iw_j}}=\frac{1}{\sqrt{w_iw_j}}.
\]
Thus
\[
\frac{d\mu_{v_0,A^{(e)}}}{d\mu_{v_0,A}}(w)
=
\frac{Z_{v_0,A}}{Z_{v_0,A^{(e)}}}\,Y_e(w).
\]
Integrating gives \(\E_{\mu_{v_0,A}}[Y_e(W)]=\frac{Z_{v_0,A^{(e)}}}{Z_{v_0,A}}\), while multiplying by \(\log Y_e(w)\) before integrating gives
\[
\E_{\mu_{v_0,A}}[Y_e(W)\log Y_e(W)]
=
\E_{\mu_{v_0,A}}[Y_e(W)]\,
\E_{\mu_{v_0,A^{(e)}}}[\log Y_e(W)].
\]

The first factor 
on the RHS
is computed directly from \eqref{eq:Zv0A}:
\[
\E_{\mu_{v_0,A}}[Y_e(W)]
=
\frac{a_e}{2}\cdot
\frac{\Gamma(b_i)}{\Gamma(b_i+\tfrac12)}\cdot
\frac{\Gamma(b_j)}{\Gamma(b_j+\tfrac12)}.
\]
For the second factor, the dependence of the density \eqref{eq:mixing_measure_magic} on \(a_e\) is exactly through the factor \(Y_e(w)^{a_e}\), so
\[
\frac{\partial}{\partial a_e}\log Z_{v_0,A}
=
\E_{\mu_{v_0,A}}[\log Y_e(W)].
\]
Differentiating the explicit formula \eqref{eq:Zv0A} then gives
\[
\E_{\mu_{v_0,A}}[\log Y_e(W)]
=
\Psi(a_e)-\frac12\Psi(b_i)-\frac12\Psi(b_j)-\log 2.
\]
Applying this identity at the shifted parameter \(A^{(e)}\) gives
\[
\E_{\mu_{v_0,A^{(e)}}}[\log Y_e(W)]
=
\Psi(a_e+1)-\frac12\Psi\bigl(b_i+\tfrac12\bigr)-\frac12\Psi\bigl(b_j+\tfrac12\bigr)-\log 2.
\]
Substituting these identities into the previous bound yields \eqref{eq:r_upper_final_clean}.
\end{proof}

\begin{remark}
The exact formula \eqref{eq:annealed_entropy_rate_integral} shows that the entropy rate is determined by the random environment through the ratios \(w_{ij}/\sqrt{w_iw_j}\), while the upper bound \eqref{eq:r_upper_final_clean} 
provides an estimate for
this dependence explicitly in terms of the initial weights.
\end{remark}

\section{KL Divergence Between Environment Laws}
\label{sec:kl_env}

Throughout this section, fix a finite connected graph \(G=(V,E)\), a distinguished initial vertex \(v_0\in V\), and a reference edge \(e_0\in E\). We work throughout in the pinned gauge \(w_{e_0}=1\). For \(A=(a_e)_{e\in E}\in(0,\infty)^E\), let \(\mu_{v_0,A}\) be the ERRW mixing measure from Definition \ref{def:magic_measure}, 
let \(Z_{v_0,A}\) be the corresponding normalizing constant, and write \(\Phi(A):=\log Z_{v_0,A}\). As before, for \(v\in V\) we write \(a_v:=\sum_{e\ni v} a_e\) and \(b_v=b_v(A):=\frac{a_v+1-\1(v=v_0)}{2}\).

For each edge \(e=\{u,v\}\in E\), define \(Y_e(w):=\frac{w_e}{\sqrt{w_uw_v}}\) and \(T_e(w):=\log Y_e(w)\). Thus \(T_e\) is exactly the logarithm of the quantity that appeared already in \cref{sec:entropy}.

The starting point is that, once the graph and the initial vertex are fixed, the family \(\{\mu_{v_0,A}\}_{A\in(0,\infty)^E}\) forms a canonical exponential family 
in the standard sense of exponential-family theory; see, for example,
\cite{Brown1986,WainwrightJordan2008}.
As a consequence, the KL divergence between two environment laws is determined by the log-normalizing constant 
\(\Phi(A)\), where $A$ denotes the vector of initial edge weights.

\subsection{Explicit formula for environment-level KL divergence}

\begin{theorem}[Exponential-family structure and explicit KL formula]
\label{thm:env_kl_main}
The family \(\{\mu_{v_0,A}\}_{A\in(0,\infty)^E}\) is a regular canonical exponential family on the pinned space \(\mathcal W_{e_0}:=\{w\in(0,\infty)^E:\ w_{e_0}=1\}\), with natural parameter \(A\), sufficient statistic \(T=(T_e)_{e\in E}\), and log-partition function \(\Phi\). More precisely,
\begin{equation}
\label{eq:exp_family_magic}
d\mu_{v_0,A}(w)
=
\exp\Big(\sum_{e\in E} a_e T_e(w)-\Phi(A)\Big)\,\nu(dw),
\end{equation}
where the base measure \(\nu\) does not depend on \(A\).
Consequently, for every \(e,e'\in E\),
\begin{equation}
\label{eq:grad_hess_env}
\frac{\partial}{\partial a_e}\Phi(A)=\E_{\mu_{v_0,A}}[T_e(W)],
\qquad
\frac{\partial^2}{\partial a_e\partial a_{e'}}\Phi(A)
=
\Cov_{\mu_{v_0,A}}\bigl(T_e(W),T_{e'}(W)\bigr).
\end{equation}
If \(A^{(0)},A^{(1)}\in(0,\infty)^E\), and \(\mu_s:=\mu_{v_0,A^{(s)}}\) for \(s\in\{0,1\}\), 
then the KL divergence is the Bregman divergence of \(\Phi\):
\begin{equation}
\label{eq:KL_bregman_env}
\KL(\mu_0\|\mu_1)
=
\Phi(A^{(1)})-\Phi(A^{(0)})
-\sum_{e\in E}(a_e^{(1)}-a_e^{(0)})\,
\frac{\partial}{\partial a_e}\Phi(A)\Big|_{A=A^{(0)}}.
\end{equation}
Define, for \(p,q>0\), 
$\Lambda(p,q):=\log\Gamma(q)-\log\Gamma(p)-(q-p)\Psi(p)$.
Then, with \(b_v^{(s)}:=b_v(A^{(s)})\),
\begin{equation}
\label{eq:KL_env_closed_form}
\KL(\mu_0\|\mu_1)
=
\sum_{e\in E}\Lambda(a_e^{(0)},a_e^{(1)})
-
\sum_{v\in V}\Lambda(b_v^{(0)},b_v^{(1)}).
\end{equation}
\end{theorem}

\begin{proof}
Starting from \eqref{eq:mixing_measure_magic}, the density of \(\mu_{v_0,A}\) is
\[
d\mu_{v_0,A}(w)
=
Z_{v_0,A}^{-1}\,
\frac{w_{v_0}^{1/2}\prod_{e\in E} w_e^{a_e-1}}
{\prod_{v\in V} w_v^{(a_v+1)/2}}
\sqrt{\tau(w)}\,dw_{-e_0}.
\]
For each edge \(e=\{u,v\}\), the dependence on \(a_e\) is exactly through the factor
\(w_e^{a_e}w_u^{-a_e/2}w_v^{-a_e/2}=Y_e(w)^{a_e}=e^{a_eT_e(w)}\). Thus
\[
d\mu_{v_0,A}(w)
=
\exp\Big(\sum_{e\in E} a_e T_e(w)-\Phi(A)\Big)\,\nu(dw),
\]
where
\[
\nu(dw):=
w_{v_0}^{1/2}
\Bigl(\prod_{e\in E} w_e^{-1}\Bigr)
\Bigl(\prod_{v\in V} w_v^{-1/2}\Bigr)
\sqrt{\tau(w)}\,dw_{-e_0}
\]
does not depend on \(A\). This proves \eqref{eq:exp_family_magic}.

The identities \eqref{eq:grad_hess_env} and \eqref{eq:KL_bregman_env} are the
standard gradient, Hessian, and KL--Bregman identities for regular canonical
exponential families; see, for example, \cite{Brown1986,WainwrightJordan2008}.
Indeed, differentiating the log-partition function in \eqref{eq:exp_family_magic}
gives the mean and covariance identities in \eqref{eq:grad_hess_env}, while direct
substitution of \eqref{eq:exp_family_magic} into
\(\KL(\mu_0\|\mu_1)\) gives the Bregman formula
\eqref{eq:KL_bregman_env}.



To verify \eqref{eq:KL_env_closed_form}, we compute the following explicit formulas. By \eqref{eq:Zv0A}, we have
\[
\Phi(A)
=
\sum_{e\in E}\log\Gamma(a_e)
-\sum_{v\in V}\log\Gamma(b_v)
+\frac{|V|-1}{2}\log\pi
-\Bigl(1-|V|+\sum_{e\in E} a_e\Bigr)\log 2.
\]
Differentiating with respect to \(a_e\), where \(e=\{u,v\}\), gives 
\begin{equation}
\label{eq:grad_logZ_env_explicit}
\frac{\partial}{\partial a_e}\Phi(A)
=
\Psi(a_e)-\frac12\Psi(b_u)-\frac12\Psi(b_v)-\log 2.
\end{equation}
Substituting these expressions in \eqref{eq:KL_bregman_env} yields \eqref{eq:KL_env_closed_form}.

\end{proof}

\begin{remark}
The exponential-family structure in \eqref{eq:exp_family_magic} is already implicit
in the magic formula and in the Bayesian interpretation of ERRW as a prior on
reversible Markov chains; see
\cite{CoppersmithDiaconis1986,KeaneRolles2000,Rolles2003,DiaconisRolles2006,MerklOryRolles2008}.
The point of \Cref{thm:env_kl_main} is to make explicit the associated
Bregman-divergence formula for the KL divergence between two ERRW environment laws,
and to record the resulting closed form in terms of Gamma and digamma functions.
\end{remark}

\subsection{A probabilistic interpretation via the STZ field}

The formula \eqref{eq:KL_env_closed_form} also admits a natural probabilistic interpretation through the random-field representation developed by Sabot, Tarr\`es, and Zeng \cite{SabotTarres2015,SabotTarresZeng2017}. 
In the original STZ field representation, the pinned field includes an
auxiliary ``ghost'' component: the field at $v_0$ has an extra independent
$\Gam(\frac12,1)$ factor. 
We consider a version of the field with this component removed,
and call the resulting reduced representation the
\emph{deghosted} parametrization.
In this parametrization, 
the underlying field separates into an ERRW environment component together with independent Gamma variables on edges and independent Gamma variables on vertices. The KL identity then appears as the net contribution of these two Gamma structures.

Let \(\beta=(\beta_e)_{e\in E}\) have independent coordinates with \(\beta_e\sim \Gam(a_e,1)\). Conditional on \(\beta\), and in the gauge \(\phi_{v_0}=0\), let \(\phi=(\phi_i)_{i\in V}\) be sampled according to the hyperbolic Gaussian density
\[
p_\beta(\phi)
=
\frac{1}{(2\pi)^{(n-1)/2}}
e^{-\sum_{e=\{i,j\}\in E}\beta_e(\cosh(\phi_i-\phi_j)-1)}
e^{-\sum_{i\in V\setminus\{v_0\}}\phi_i}
\sqrt{\sum_{T\in\mathcal T}\prod_{e=\{i,j\}\in T}\beta_e e^{\phi_i+\phi_j}}.
\]
We then define, for \(i\in V\) and \(\{i,j\}\in E\),
\(S_i:=\frac12\sum_{j\sim i}\beta_{ij}e^{\phi_j-\phi_i}\) and
\(W_{ij}:=\beta_{ij}e^{\phi_i+\phi_j}\).
Since \(W_i:=\sum_{j\sim i}W_{ij}=2S_i e^{2\phi_i}\), it is natural to introduce the normalized field
\[
\widetilde\beta_{ij}
=
\frac{\beta_{ij}}{\sqrt{S_iS_j}}
=
\frac{2W_{ij}}{\sqrt{W_iW_j}}.
\]

The first point is that the vertex variables separate from the environment.

\begin{lemma}[Deghosted vertex field~\cite{dingananth26}]
\label{lem:deghosted_S_gamma}
The coordinates of \(S\) are independent, with \(S_v\sim \Gam(b_v,1)\) for \(v\in V\).
\end{lemma}

We next record the basic structural properties of the deghosted parametrization.

\begin{lemma}[Deghosted parametrization, normalized field, and factorization]
\label{lem:deghosted_structure}
Fix the gauge \(\phi_{v_0}=0\).

\begin{enumerate}
    \item \emph{Bijection between \((\beta,\phi)\) and \((W,S)\).}
    The map \((\beta,\phi)\mapsto (W,S)\) is measurable and one-to-one onto its image. Its inverse is given by
    \[
    \phi_i=\frac12\log\frac{W_i}{2S_i},
    \qquad i\in V,
    \]
    and \(\beta_{ij}=W_{ij}e^{-(\phi_i+\phi_j)}\).

    \item \emph{Pinned environments and normalized fields.}
    Let \(\mathcal W_{e_0}:=\{w\in(0,\infty)^E:\ w_{e_0}=1\}\), and define
    \[
    H:\mathcal W_{e_0}\to(0,\infty)^E,
    \qquad
    H(w)_{ij}:=\frac{2w_{ij}}{\sqrt{w_iw_j}},
    \qquad \{i,j\}\in E.
    \]
    Then \(H\) is injective. Moreover, a vector \(\widetilde\beta\in(0,\infty)^E\) lies in the image of \(H\) if and only if 
     the matrix \(\widetilde B\) with rows and column indexed by the vertices and with entries \(\widetilde\beta_{ij}\)
    has spectral radius \(2\). In that case, if \(u>0\) satisfies \(\widetilde B u=2u\) and is normalized by \(\frac12\,\widetilde\beta_{i_0j_0}u_{i_0}u_{j_0}=1\), then the unique \(w\in\mathcal W_{e_0}\) with \(H(w)=\widetilde\beta\) is given by \(w_{ij}=\frac12\,\widetilde\beta_{ij}u_i u_j\).

    \item \emph{Factorization.}
    Let \(P_A\) denote the law of \((\beta,\phi)\) under parameter \(A\). Then the pushforward of \(P_A\) under \((\beta,\phi)\mapsto (W,S)\) factorizes as
    \begin{equation}
    \label{eq:factorization_W_S}
    (W,S)\sim \mu_{v_0,A}\otimes \bigotimes_{v\in V}\Gam(b_v,1).
    \end{equation}
    Equivalently, if \(\widetilde\nu_{v_0,A}:=H_\#\mu_{v_0,A}\) denotes the law of the normalized field, then
    \begin{equation}
    \label{eq:factorization_tildebeta_S}
    (\widetilde\beta,S)\sim \widetilde\nu_{v_0,A}\otimes \bigotimes_{v\in V}\Gam(b_v,1).
    \end{equation}
\end{enumerate}
\end{lemma}

\begin{proof}
For part (1), the identities \(W_i=\sum_{j\sim i}W_{ij}=2S_i e^{2\phi_i}\), \(i\in V\), determine \(\phi\) uniquely from \((W,S)\), since \(\phi_{v_0}=0\) is fixed. Once \(\phi\) is known, the variables \(\beta_{ij}\) are recovered from \(W_{ij}=\beta_{ij}e^{\phi_i+\phi_j}\). This proves the claimed bijection.

For part (2), let \(w\in\mathcal W_{e_0}\), let \(W\) be 
the matrix with rows and columns indexed by the vertices and with the entries $W_{ij}$,
and let \(D:=\diag(w_i:i\in V)\). Then the 
matrix 
associated with \(H(w)\) is 
\(\widetilde B=2D^{-1/2}WD^{-1/2}\). 
If \(P:=D^{-1}W\), then \(P\) is the transition matrix of the reversible Markov chain with conductances \(w\), and \(\widetilde B=2D^{1/2}PD^{-1/2}\). Thus \(\widetilde B\) is similar to \(2P\). Since \(G\) is connected and all edge weights are positive, \(P\) is irreducible and stochastic, so \(\rho(P)=1\). Hence \(\rho(\widetilde B)=2\).

Now let \(u_i:=\sqrt{w_i}\). Then
\[
(\widetilde B u)_i
=
\sum_{j\sim i}\frac{2w_{ij}}{\sqrt{w_iw_j}}\sqrt{w_j}
=
\frac{2}{\sqrt{w_i}}\sum_{j\sim i} w_{ij}
=
2\sqrt{w_i}
=
2u_i.
\]
Thus \(u\) is a positive eigenvector of \(\widetilde B\) for eigenvalue \(2\). By 
the
Perron--Frobenius
theorem, 
such an eigenvector is unique up to scale. Since the pinning condition \(w_{e_0}=1\) fixes the scale, \(H\) is injective.

Conversely, let \(\widetilde\beta\in(0,\infty)^E\) be such that 
the matrix \(\widetilde B\) with rows and columns indexed by the vertices and with entries \(\widetilde\beta_{ij}\)
satisfies \(\rho(\widetilde B)=2\). By 
the
Perron--Frobenius
theorem,
there exists \(u>0\) with \(\widetilde B u=2u\). Define \(w_{ij}:=\frac12\,\widetilde\beta_{ij}u_i u_j\). Then
\[
w_i=\sum_{j\sim i} w_{ij}
=\frac12\,u_i\sum_{j\sim i}\widetilde\beta_{ij}u_j
=\frac12\,u_i(\widetilde B u)_i
=u_i^2.
\]
Hence
\[
H(w)_{ij}
=
\frac{2w_{ij}}{\sqrt{w_iw_j}}
=
\frac{2\cdot \frac12\,\widetilde\beta_{ij}u_i u_j}{u_i u_j}
=
\widetilde\beta_{ij}.
\]
If \(u\) is normalized by \(\frac12\,\widetilde\beta_{i_0j_0}u_{i_0}u_{j_0}=1\), then \(w_{e_0}=1\), so \(w\in\mathcal W_{e_0}\). This proves the image characterization and reconstruction formula.

For part (3), part (1) shows that the change of variables from \((\beta,\phi)\) to \((W,S)\) is bijective onto its image. The STZ representation implies that, after this change of variables, the \(W\)-component has law \(\mu_{v_0,A}\), while the coordinates of \(S\) are independent with laws \(\Gam(b_v,1)\); see \cite{SabotTarres2015,SabotTarresZeng2017}. This proves \eqref{eq:factorization_W_S}. Applying the map \(H\) to the \(W\)-coordinate yields \eqref{eq:factorization_tildebeta_S}.
\end{proof}

We may now recover the KL identity from the deghosted factorization.

\begin{theorem}[Probabilistic interpretation of the environment KL formula]
\label{thm:env_kl_info_theoretic}
Let \(A^{(0)},A^{(1)}\in(0,\infty)^E\), and let \(\mu_s:=\mu_{v_0,A^{(s)}}\), \(\widetilde\nu_s:=H_\#\mu_{v_0,A^{(s)}}\), and \(b_v^{(s)}:=b_v(A^{(s)})\), for \(s\in\{0,1\}\). Then
\begin{equation}
\label{eq:env_kl_info_theoretic}
\KL(\mu_0\|\mu_1)
=
\KL(\widetilde\nu_0\|\widetilde\nu_1)
=
\sum_{e\in E}\Lambda(a_e^{(0)},a_e^{(1)})
-
\sum_{v\in V}\Lambda(b_v^{(0)},b_v^{(1)}).
\end{equation}
In particular, the environment-level KL divergence is the net contribution of the edge-Gamma and vertex-Gamma parts of the deghosted field.
\end{theorem}

\begin{proof}
Let \(P_s\) be the law of \((\beta,\phi)\) under parameter \(A^{(s)}\), for \(s\in\{0,1\}\). Since the conditional law of \(\phi\) given \(\beta\) does not depend on \(A\), the chain rule for relative entropy gives
\[
\KL(P_0\|P_1)
=
\KL(\beta^{(0)}\|\beta^{(1)})
=
\sum_{e\in E}\Lambda(a_e^{(0)},a_e^{(1)}),
\]
because the coordinates \(\beta_e\) are independent Gamma variables.

On the other hand, by Lemma \ref{lem:deghosted_structure}(1), relative entropy is preserved under the bijection \((\beta,\phi)\mapsto (W,S)\). By Lemma \ref{lem:deghosted_structure}(3), the pushforward law factorizes as
\[
(W,S)\sim \mu_{v_0,A}\otimes \bigotimes_{v\in V}\Gam(b_v,1).
\]
Hence
\[
\KL(P_0\|P_1)
=
\KL(\mu_0\|\mu_1)
+
\sum_{v\in V}\Lambda(b_v^{(0)},b_v^{(1)}).
\]
Comparing the two expressions yields
\[
\KL(\mu_0\|\mu_1)
=
\sum_{e\in E}\Lambda(a_e^{(0)},a_e^{(1)})
-
\sum_{v\in V}\Lambda(b_v^{(0)},b_v^{(1)}).
\]

Finally, by Lemma \ref{lem:deghosted_structure}(2), the map \(H\) is a bijection between the pinned environment and the normalized field. Since relative entropy is invariant under measurable bijections,
\[
\KL(\mu_0\|\mu_1)=\KL(\widetilde\nu_0\|\widetilde\nu_1).
\]
This proves \eqref{eq:env_kl_info_theoretic}.
\end{proof}

In light of this, the environment-level KL divergence should be understood as
\[
\KL(\mu_0\|\mu_1)
=
\KL(\beta^{(0)}\|\beta^{(1)})
-
\KL(S^{(0)}\|S^{(1)})
=
\sum_{e\in E}\KL(\beta_e^{(0)}\|\beta_e^{(1)})
-
\sum_{v\in V}\KL(S_v^{(0)}\|S_v^{(1)}).
\]
Indeed, under the deghosted factorization, the raw field splits into an environment component and an independent vertex Gamma field, while the edge variables \((\beta_e)_{e\in E}\) are themselves independent Gamma variables. The ghost variable, when included, has the same law under both parameters and therefore contributes no relative entropy. Thus the KL divergence of the environment is obtained by subtracting the vertex-Gamma contribution from the total KL divergence of the independent edge-Gamma field.

\section{KL Divergence Between Trajectory Laws}
\label{sec:kl_path}

We now turn from the hidden environment to the observable trajectory. For \(T\ge 1\) and \(s\in\{0,1\}\), let
\[
P_s^{(T)}:=\mathcal L_{\Perrw^{v_0,A^{(s)}}}(X_0^T)
\]
denote the law of the first \(T\) steps of ERRW with initial weights \(A^{(s)}\). These trajectory laws are the natural objects in statistical testing problems based on finite observed paths. In particular, if one observes i.i.d.\ \(T\)-step trajectories drawn either from \(P_0^{(T)}\) or from \(P_1^{(T)}\), then Stein's lemma identifies \(\KL(P_0^{(T)}\|P_1^{(T)})\) as the optimal type-II error exponent at fixed type-I error level. The main point of this section is that the difference between the environment-level and trajectory-level KL divergences is exactly an expected posterior KL divergence.

\subsection{Trajectory laws, testing, and posterior conjugacy}

The relevant posterior update is especially simple for ERRW: after observing a path, one adds the undirected edge counts to the parameters and moves the distinguished root to the terminal vertex. This is the conjugacy property underlying the Bayesian approach of Diaconis and Rolles \cite{DiaconisRolles2006}.

\begin{proposition}[Posterior conjugacy and path probabilities]
\label{prop:posterior_conjugacy}
Fix \(A\in(0,\infty)^E\), let \(W\sim\mu_{v_0,A}\), and, conditional on \(W=w\), let \(X_0^T\) be sampled from the quenched chain \(\Psrw^{v_0,w}\). For a realized path \(x_0^T\) with \(x_0=v_0\), let \(N=N(x_0^T)=(N_e)_{e\in E}\) be its undirected edge-count vector and let \(x_T\) be its endpoint. Then
\begin{equation}
\label{eq:unnormalized_bayes_magic_main}
\Psrw^{v_0,w}(X_0^T=x_0^T)\,\mu_{v_0,A}(dw)
=
\frac{Z_{x_T,A+N}}{Z_{v_0,A}}\,
\mu_{x_T,A+N}(dw).
\end{equation}
Consequently,
\begin{equation}
\label{eq:path_probability_Z_ratio_main}
\Perrw^{v_0,A}(X_0^T=x_0^T)=\frac{Z_{x_T,A+N}}{Z_{v_0,A}},
\end{equation}
and
\begin{equation}
\label{eq:posterior_path_main}
\mu_{v_0,A}(dw\mid X_0^T=x_0^T)=\mu_{x_T,A+N}(dw).
\end{equation}
Moreover, since the endpoint is determined by the parity rule
\eqref{eq:parityrule},
the posterior depends on the path only through \(N\):
\begin{equation}
\label{eq:posterior_count_main}
\mu_{v_0,A}(dw\mid N)=\mu_{x_T(N),A+N}(dw).
\end{equation}
\end{proposition}

\begin{proof}
See Appendix~\ref{sec:proof_posterior_conjugacy}.
\end{proof}

\subsection{The posterior-gap identity and general bounds}

We now compare the environment-level and trajectory-level KL divergences. Let
\[
\mu_s:=\mu_{v_0,A^{(s)}},
\qquad
P_s^{(T)}:=\mathcal L_{\Perrw^{v_0,A^{(s)}}}(X_0^T),
\qquad s\in\{0,1\},
\]
and define the gap
\[
\Gap_T:=\KL(\mu_0\|\mu_1)-\KL(P_0^{(T)}\|P_1^{(T)}).
\]

\begin{theorem}[Posterior-gap identity and explicit formula]
\label{thm:gap_formula}
For every \(T\ge 0\),
\begin{equation}
\label{eq:gap_conditional_KL_main}
\Gap_T
=
\E_{P_0^{(T)}}\!\left[
\KL\bigl(P_0(W\mid N)\,\|\,P_1(W\mid N)\bigr)
\right].
\end{equation}
Equivalently, if \(N=N(X_0^T)\) and \(x=x_T(N)\), then
\begin{equation}
\label{eq:gap_posterior_env_main}
\Gap_T
=
\E_{P_0^{(T)}}\!\left[
\KL\bigl(\mu_{x,A^{(0)}+N}\,\|\,\mu_{x,A^{(1)}+N}\bigr)
\right].
\end{equation}

For \(s\in\{0,1\}\), write \(a_v^{(s)}:=\sum_{e\ni v}a_e^{(s)}\), and let \(d_v(N):=\sum_{e\ni v}N_e\). Define
\[
b_v^{(s,N)}
:=
\frac{a_v^{(s)}+d_v(N)+1-\1(v=x_T(N))}{2},
\qquad v\in V.
\]
Then
\begin{equation}
\label{eq:gap_explicit_main}
\Gap_T
=
\E_{P_0^{(T)}}\!\left[
\sum_{e\in E}\Lambda(a_e^{(0)}+N_e,a_e^{(1)}+N_e)
-
\sum_{v\in V}\Lambda(b_v^{(0,N)},b_v^{(1,N)})
\right].
\end{equation}
\end{theorem}

\begin{proof}
Let \(Q_s\) be the joint law of \((W,X_0^T)\) under parameter \(A^{(s)}\). Since the conditional law of \(X_0^T\) given \(W=w\) is the same quenched law \(\Psrw^{v_0,w}\) under both parameters, the chain rule for relative entropy gives
\[
\KL(Q_0\|Q_1)=\KL(\mu_0\|\mu_1).
\]
Decomposing in the opposite order yields
\[
\KL(Q_0\|Q_1)
=
\KL(P_0^{(T)}\|P_1^{(T)})
+
\E_{P_0^{(T)}}\!\left[
\KL\bigl(P_0(W\mid X_0^T)\,\|\,P_1(W\mid X_0^T)\bigr)
\right].
\]
Subtracting the two identities proves
\[
\Gap_T
=
\E_{P_0^{(T)}}\!\left[
\KL\bigl(P_0(W\mid X_0^T)\,\|\,P_1(W\mid X_0^T)\bigr)
\right].
\]
By Proposition \ref{prop:posterior_conjugacy}, \(P_s(W\mid X_0^T)=P_s(W\mid N)=\mu_{x,A^{(s)}+N}\), where \(x=x_T(N)\). This gives \eqref{eq:gap_conditional_KL_main} and \eqref{eq:gap_posterior_env_main}.

It remains to apply the environment-level KL formula from \Cref{thm:env_kl_main}. For the posterior measure \(\mu_{x,A^{(s)}+N}\), the edge parameters are \(a_e^{(s)}+N_e\), while the corresponding vertex parameters are
\[
\frac{\sum_{e\ni v}(a_e^{(s)}+N_e)+1-\1(v=x)}{2}
=
\frac{a_v^{(s)}+d_v(N)+1-\1(v=x)}{2}
=
b_v^{(s,N)}.
\]
Substituting these into \eqref{eq:KL_env_closed_form} yields \eqref{eq:gap_explicit_main}.
\end{proof}

The exact expression \eqref{eq:gap_explicit_main} is convenient because the vertex term is nonpositive, so one can obtain upper bounds by estimating only the edge contribution. The next proposition records both an integral representation and a simple upper bound.

\begin{proposition}[General bounds on \(\Gap_T\)]
\label{prop:gap_general_bounds}
Let \(\delta_e:=a_e^{(1)}-a_e^{(0)}\), let \(\underline a:=\min_{e\in E}\min\{a_e^{(0)},a_e^{(1)}\}\), and let \(\Psi_1=\Psi'\) be the trigamma function. Then
\begin{align}
\Gap_T
&=
\int_0^1 (1-t)\Bigg(
\sum_{e\in E}\delta_e^2\,
\E_{P_0^{(T)}}\!\Big[\Psi_1\!\bigl(a_e^{(0)}+N_e+t\delta_e\bigr)\Big]
\notag\\
&\hspace{6em}
-\frac14\sum_{v\in V}\delta_v^2\,
\E_{P_0^{(T)}}\!\Big[\Psi_1\!\bigl(b_v^{(0,N)}+\tfrac{t\delta_v}{2}\bigr)\Big]
\Bigg)\,dt,
\label{eq:gap_integral_main}
\end{align}
where \(\delta_v:=a_v^{(1)}-a_v^{(0)}=\sum_{e\ni v}\delta_e\).

In particular,
\begin{equation}
\label{eq:gap_upper_main}
\Gap_T
\le
\frac12\sum_{e\in E}\delta_e^2\,
\E_{P_0^{(T)}}\!\left[
\frac{1}{N_e+\underline a}
+
\frac{1}{(N_e+\underline a)^2}
\right].
\end{equation}
Hence also
\begin{equation}
\label{eq:gap_upper_simplified_main}
\Gap_T
\le
\frac12\Bigl(1+\frac1{\underline a}\Bigr)
\sum_{e\in E}\delta_e^2\,
\E_{P_0^{(T)}}\!\left[\frac{1}{N_e+\underline a}\right].
\end{equation}
\end{proposition}

\begin{proof}
We use the elementary identity
\[
\Lambda(p,q)
=
(q-p)^2\int_0^1 (1-t)\,\Psi_1\!\bigl(p+t(q-p)\bigr)\,dt,
\qquad p,q>0,
\]
which follows from Taylor's theorem with integral remainder applied to \(\log\Gamma\). Applying this to each term in \eqref{eq:gap_explicit_main} yields \eqref{eq:gap_integral_main}.

For the upper bound, discard the nonpositive vertex term in \eqref{eq:gap_explicit_main}. Since \(\Psi_1(x)\le x^{-1}+x^{-2}\) for \(x>0\), we obtain
\[
\Lambda(a_e^{(0)}+N_e,a_e^{(1)}+N_e)
\le
\frac{\delta_e^2}{2}\left(
\frac{1}{m_e(N)}+\frac{1}{m_e(N)^2}
\right),
\]
where \(m_e(N):=\min\{a_e^{(0)}+N_e,a_e^{(1)}+N_e\}\). Since \(m_e(N)\ge N_e+\underline a\), this gives \eqref{eq:gap_upper_main} after taking expectation. Finally,
\[
\frac{1}{(N_e+\underline a)^2}\le \frac1{\underline a}\cdot \frac{1}{N_e+\underline a},
\]
which implies \eqref{eq:gap_upper_simplified_main}.
\end{proof}

\subsection{The \(n\)-star case}    \label{sec:n-star-asyptotics}

We next consider the \(n\)-star, in which ERRW reduces to a P\'olya urn. This gives a tractable model in which the inverse local-time estimates underlying the KL gap can be analyzed sharply.

Let \(G\) be the \(n\)-star with center \(v_0\) and \(n\) leaves, and suppose that all initial edge weights are equal to \(a>0\). Fix one leaf edge \(e\). Since each excursion from the center consists of traversing exactly one leaf edge out and then back, it is natural to index the process by excursions from \(v_0\). For \(T\ge 1\), let \(N_e(T)\),
with a harmless abuse of notation,
denote the number of the first \(T\) excursions that use the edge \(e\).

\begin{theorem}[Inverse local-time bound on the \(n\)-star]
\label{thm:nstar_inverse_moment}
Fix \(\gamma>0\). Then starting at the center node $v_0$, as \(T\to\infty\),
\[
\Eerrw^{v_0, A}\!\left[\frac{1}{N_e(T)+\gamma}\right]
=
\begin{cases}
O(T^{-1}), & a>1,\\[3pt]
O((\log T)\,T^{-1}), & a=1,\\[3pt]
O(T^{-a}), & 0<a<1,
\end{cases}
\]
with constants depending only on \(a\), \(n\), and \(\gamma\).
\end{theorem}

The proof is based on the beta-binomial representation of \(N_e(T)\) together with a sharp estimate on the binomial inverse moment \(g_T(x):=\E[(\Bin(T,x)+\gamma)^{-1}]\); see Appendix~\ref{sec:proof_nstar_inverse}.

Combining \Cref{thm:nstar_inverse_moment} with Proposition \ref{prop:gap_general_bounds} yields the corresponding decay rates for the trajectory-level KL gap.

\begin{corollary}[Gap decay on the \(n\)-star]
\label{cor:nstar_gap_decay}
Let \(G\) be the \(n\)-star, and let \(A^{(s)}\), \(s\in\{0,1\}\), be constant edge-weight vectors with values \(a^{(s)}>0\). Then starting at the center node $v_0$, we have 
\[
\Gap_T=
\begin{cases}
O(T^{-1}), & a^{(0)},a^{(1)}>1,\\[3pt]
O((\log T)\,T^{-1}), & \min\{a^{(0)},a^{(1)}\}=1,\\[3pt]
O\!\bigl(T^{-\min\{a^{(0)},a^{(1)}\}}\bigr), & \min\{a^{(0)},a^{(1)}\}<1.
\end{cases}
\]
\end{corollary}

\begin{proof}
Apply Proposition \ref{prop:gap_general_bounds}. In the \(n\)-star with constant initial weights, all edge counts are identically distributed under \(P_0^{(T)}\), so the right-hand side of \eqref{eq:gap_upper_simplified_main} is bounded by a constant multiple of \(\E[(N_e(T)+\gamma)^{-1}]\), for any fixed \(\gamma>0\). The stated rates then follow from \Cref{thm:nstar_inverse_moment}.
\end{proof}

\subsection{The general graph case}

We 
carry out 
the proof in three steps,
followed by a step to try to optimize the bounds.

\paragraph{Step 1: tails of the random environment.}

We begin by controlling the lower tail of the normalized environment
\[
X_e:=\frac{W_e}{\sum_{g\in E}W_g},
\qquad e\in E,
\]
where \(W\sim \mu_{v_0,A}\). By \cite{MerklRolles2008Bounding}, the law of \(X=(X_g)_{g\in E}\) has a density on the open simplex
\[
\Delta_E:=\Bigl\{x=(x_g)_{g\in E}\in(0,1)^E:\ \sum_{g\in E}x_g=1\Bigr\}
\]
of the form
\begin{equation}
\label{eq:simplex_density_general_graph}
\rho_{v_0,A}(x)
=
\widetilde Z_{v_0,A}^{-1}
\frac{\prod_{g\in E}x_g^{a_g-1}}
{\prod_{v\in V}x_v^{b_v}}
\sqrt{\tau(x)},
\qquad x\in \Delta_E,
\end{equation}
where \(x_v:=\sum_{g\ni v}x_g\), \(b_v:=\frac{a_v+1-\1(v=v_0)}{2}\), \(\tau(x):=\sum_{T\in\mathcal T}\prod_{g\in T}x_g\), and \(\widetilde Z_{v_0,A}=\Gamma(|E|)\cdot Z_{v_0, A}\) is the normalizing constant.

For a nonempty set \(F\subseteq E\), let \(a(F):=\sum_{g\in F}a_g\), and let
\[
S(F):=\{v\in V:\ \text{every edge incident to }v\text{ belongs to }F\}.
\]
We also write \(\kappa(G\setminus F)\) for the number of connected components of the graph obtained by deleting the edges in \(F\), and define
\begin{equation}
\label{eq:AF_definition}
A(F):=
a(F)-\sum_{v\in S(F)} b_v+\frac{\kappa(G\setminus F)-1}{2}.
\end{equation}

The key combinatorial input is the following fact.

\begin{lemma}
\label{lem:A-lower-bound-explicit}
Assume that \(a_g\ge \underline a>0\) for all \(g\in E\). Then, for every nonempty proper subset \(F\subsetneq E\), one has
\[
A(F)\ge \frac{\underline a}{2}.
\]
\end{lemma}

\begin{proof}
Deferred to Appendix~\ref{sec:proof_A_lower_bound}.
Note that the claim does not hold if $F = \emptyset$ or if $F = E$, because $A(F) = 0$
in both these cases.
\end{proof}

We will also make use of the following lemma.
\begin{lemma}[Ordered small-edge coordinates]
\label{lem:ordered_small_edge_coordinates}
Let \(F\subsetneq E\) be nonempty, write \(m:=|F|\), and fix an ordering \(\pi\), such that the edges in $F$ is ordered as
\(\{g_1,\dots,g_m\}\) under $\pi$. Fix $\delta_0:=(2|E|)^{-1}$, define
\begin{equation}    
\label{eq:omega_definition}
\Omega(F,\pi):=
\{x\in\Delta_E:\ 0<x_{g_1}\le \cdots \le x_{g_m}<\delta_0,\ \ x_h\ge \delta_0 \ \forall h\in E\setminus F\}.
\end{equation}
For \(x\in\Omega(F,\pi)\), 
define
\(r:=x_{g_m}\) and \(t_j:=x_{g_j}/x_{g_{j+1}}\in(0,1]\) for \(1\le j\le m-1\).
Then \(x_{g_j}=r\prod_{\ell=j}^{m-1}t_\ell\) for every \(1\le j\le m\), and
\[
dx_F
:=
\prod_{j=1}^m dx_j
=
r^{m-1}\prod_{j=1}^{m-1} t_j^{\,j-1}\,dr\,dt_1\cdots dt_{m-1}.
\]

More generally, for any real numbers \((\alpha_g)_{g\in F}\), if
\(\alpha(F_j^{\pi}):=\sum_{k=1}^j \alpha_{g_k}\), then
\[
\Bigl(\prod_{g\in F}x_g^{\alpha_g-1}\Bigr)\,dx_F
=
r^{\alpha(F)-1}\prod_{j=1}^{m-1} t_j^{\,\alpha(F_j^{\pi})-1}\,dr\,dt_1\cdots dt_{m-1}.
\]
In particular, if \(e=g_q\in F\) and \(s\ge 0\), then
\[
x_e^{-s}
=
r^{-s}\prod_{j=q}^{m-1} t_j^{-s}
=
r^{-s}\prod_{j=1}^{m-1} t_j^{-s\,\1(e\in F_j^{\pi})}.
\]
\end{lemma}

\begin{proof}
Since \(x_{g_j}=t_jx_{g_{j+1}}\) for \(1\le j\le m-1\) and \(x_{g_m}=r\), iterating gives
\(x_{g_j}=r\prod_{\ell=j}^{m-1}t_\ell\). Writing \(y_j:=x_{g_j}\), we have
\(y_m=r\) and \(y_j=t_jy_{j+1}\), so \(dy_1\cdots dy_m=(y_2\cdots y_m)\,dt_1\cdots dt_{m-1}\,dr\). Since
\(y_k=r\prod_{\ell=k}^{m-1}t_\ell\), the product \(y_2\cdots y_m\) equals
\(r^{m-1}\prod_{j=1}^{m-1}t_j^{\,j-1}\), because \(t_j\) appears in exactly \(j-1\) of the factors \(y_2,\dots,y_m\). This proves the Jacobian formula.

For the monomial identity, substitute \(x_{g_k}=r\prod_{\ell=k}^{m-1}t_\ell\) into
\(\prod_{g\in F}x_g^{\alpha_g-1}\). The exponent of \(r\) from this product is
\(\sum_{k=1}^m(\alpha_{g_k}-1)=\alpha(F)-m\), while the exponent of \(t_j\) is
\(\sum_{k=1}^j(\alpha_{g_k}-1)=\alpha(F_j^{\pi})-j\), since \(t_j\) appears exactly in
\(x_{g_1},\dots,x_{g_j}\). After multiplying by the Jacobian, these become
\(\alpha(F)-1\) and \(\alpha(F_j^{\pi})-1\), respectively, proving the second display.

Finally, if \(e=g_q\), then \(x_e=x_{g_q}=r\prod_{j=q}^{m-1}t_j\), and raising to the power \(-s\) gives the last identity. The alternative form follows because \(e=g_q\in F_j^{\pi}\) holds exactly when \(q\le j\).
\end{proof}
We may now prove the required negative-moment and tail bounds.

\begin{theorem}[Small-edge tail bound]
\label{thm:small_edge_tail}
Assume that \(a_g\ge \underline a>0\) for all \(g\in E\). Fix \(e\in E\) and \(0<s<\underline a/2\). Let \(\delta_0:=(2|E|)^{-1}\). For each nonempty proper subset \(F\subsetneq E\) with \(e\in F\), define
\[
B_{\mathrm{out}}(F):=\sum_{v\notin S(F)} b_v
\qquad\text{and}\qquad
M_{\mathrm{edge}}(F):=\prod_{h\in E\setminus F}\max\{1,\delta_0^{\,a_h-1}\}.
\]
Then
\begin{equation}
\label{eq:small_edge_moment_bound}
\E[X_e^{-s}]
\le
(2|E|)^s
+
\frac{1}{\widetilde Z_{v_0,A}}
\sum_{\substack{\varnothing\neq F\subsetneq E\\ e\in F}}
|F|!\,\sqrt{|\mathcal T|}\,M_{\mathrm{edge}}(F)\,\delta_0^{-B_{\mathrm{out}}(F)}
\frac{\delta_0^{\,\underline a/2-s}}{(\underline a/2-s)^{|F|}}.
\end{equation}
In particular, \(\E[X_e^{-s}]<\infty\). Consequently, if
\begin{equation}
\label{eq:Ce_s_def}
C_{e,s}:=
(2|E|)^s
+
\frac{1}{\widetilde Z_{v_0,A}}
\sum_{\substack{\varnothing\neq F\subsetneq E\\ e\in F}}
|F|!\,\sqrt{|\mathcal T|}\,M_{\mathrm{edge}}(F)\,\delta_0^{-B_{\mathrm{out}}(F)}
\frac{\delta_0^{\,\underline a/2-s}}{(\underline a/2-s)^{|F|}},
\end{equation}
then
\begin{equation}
\label{eq:small_edge_tail_bound}
\P(X_e<\varepsilon)\le C_{e,s}\,\varepsilon^s,
\qquad 0<\varepsilon<1.
\end{equation}
\end{theorem}

\begin{proof}
Write
\[
\E[X_e^{-s}]
=
\int_{\Delta_E} x_e^{-s}\rho_{v_0,A}(x)\,dx
=
\int_{\{x_e\ge \delta_0\}} x_e^{-s}\rho_{v_0,A}(x)\,dx
+
\int_{\{x_e<\delta_0\}} x_e^{-s}\rho_{v_0,A}(x)\,dx.
\]
On \(\{x_e\ge \delta_0\}\), one has \(x_e^{-s}\le \delta_0^{-s}=(2|E|)^s\), so
\[
\int_{\{x_e\ge \delta_0\}} x_e^{-s}\rho_{v_0,A}(x)\,dx \le (2|E|)^s.
\]

It remains to estimate the integral over \(\{x_e<\delta_0\}\). For each \(x\) in this set, let \(F(x):=\{g\in E:\ x_g<\delta_0\}\). Since \(\delta_0=(2|E|)^{-1}<|E|^{-1}\), not all coordinates of \(x\) can be \(<\delta_0\), so \(F(x)\) is always a nonempty proper subset of \(E\), and it contains \(e\). Thus
\[
\{x_e<\delta_0\}
=
\bigcup_{\substack{\varnothing\neq F\subsetneq E\\ e\in F}}
\bigcup_{\pi\in\mathfrak S(F)}
\Omega(F,\pi),
\]
where, for \(F=\{g_1,\dots,g_m\}\) listed in the order \(\pi\), $\Omega(F, \pi)$ is defined in \cref{eq:omega_definition}.
The union is disjoint up to null sets.

Fix such a pair \((F,\pi)\), write \(m:=|F|\), and let \(F_j^{\pi}:=\{g_1,\dots,g_j\}\). 
On \(\Omega(F,\pi)\), introduce the coordinates \(r:=x_{g_m}\) and \(t_j:=x_{g_j}/x_{g_{j+1}}\in(0,1]\) for \(j=1,\dots,m-1\). By Lemma~\ref{lem:ordered_small_edge_coordinates},
\[
\Bigl(\prod_{g\in F}x_g^{a_g-1}\Bigr)\,dx_F
=
r^{a(F)-1}\prod_{j=1}^{m-1} t_j^{a(F_j^{\pi})-1}\,dr\,dt_1\cdots dt_{m-1}.
\]

For the remaining factors, we use the following bounds on \(\Omega(F,\pi)\). First, for \(h\in E\setminus F\), since \(x_h\in[\delta_0,1]\), we have \(x_h^{a_h-1}\le \max\{1,\delta_0^{\,a_h-1}\}\), and hence \(\prod_{h\in E\setminus F}x_h^{a_h-1}\le M_{\mathrm{edge}}(F)\).

Second, if \(v\notin S(F)\), then some incident edge lies outside \(F\), so \(x_v\ge \delta_0\), and therefore \(x_v^{-b_v}\le \delta_0^{-b_v}\). For \(v\in S(F)\), all incident edges of \(v\) belong to \(F\). Let
\[
m_v:=\max\{k\in\{1,\dots,m\}: g_k \text{ is incident to } v\}.
\]
Then \(x_v=\sum_{e\ni v}x_e\ge x_{g_{m_v}}\), and since \(x_{g_k}=r\prod_{\ell=k}^{m-1} t_\ell\), we get
\[
x_v\ge x_{g_{m_v}}=r\prod_{\ell=m_v}^{m-1} t_\ell,
\qquad\text{so}\qquad
x_v^{-b_v}\le r^{-b_v}\prod_{\ell=m_v}^{m-1} t_\ell^{-b_v}.
\]

Multiplying these bounds over \(v\notin S(F)\) gives the factor \(\delta_0^{-B_{\mathrm{out}}(F)}\), where
\[
B_{\mathrm{out}}(F):=\sum_{v\notin S(F)} b_v.
\]
Multiplying over \(v\in S(F)\) gives the factor
\[
r^{-\sum_{v\in S(F)} b_v}
\prod_{v\in S(F)}\prod_{\ell=m_v}^{m-1} t_\ell^{-b_v}.
\]
Now fix \(j\in\{1,\dots,m-1\}\). The variable \(t_j\) appears in the contribution of \(v\) exactly when \(m_v\le j\). Since \(m_v\) is the largest index of an edge in \(F\) incident to \(v\), the condition \(m_v\le j\) is equivalent to saying that every edge incident to \(v\) lies in \(F_j^{\pi}=\{g_1,\dots,g_j\}\), i.e.\ \(v\in S(F_j^{\pi})\). Indeed, we show both implications.

\emph{(\(\Rightarrow\))} Suppose \(m_v\le j\). By definition, \(m_v\) is the largest index \(k\) such that \(g_k\) is incident to \(v\). Hence no edge \(g_k\) with \(k>j\) is incident to \(v\). Since \(v\in S(F)\), every edge incident to \(v\) belongs to \(F\), so every edge incident to \(v\) must in fact belong to \(F_j^{\pi}=\{g_1,\dots,g_j\}\). Therefore \(v\in S(F_j^{\pi})\).

\emph{(\(\Leftarrow\))} Conversely, suppose \(v\in S(F_j^{\pi})\). Then every edge incident to \(v\) lies in \(F_j^{\pi}\). In particular, if \(g_k\) is incident to \(v\), then \(k\le j\). Since \(m_v\) is the largest index of an edge in \(F\) incident to \(v\), it follows that \(m_v\le j\).

Thus
\[
m_v\le j
\quad\Longleftrightarrow\quad
v\in S(F_j^{\pi}).
\]
Therefore the total exponent of \(t_j\) is \(-\sum_{v\in S(F):\,m_v\le j} b_v = -\sum_{v\in S(F_j^{\pi})} b_v\). Combining everything, we obtain
\[
\prod_{v\in V}x_v^{-b_v}
\le
\delta_0^{-B_{\mathrm{out}}(F)}
\,r^{-\sum_{v\in S(F)}b_v}
\prod_{j=1}^{m-1} t_j^{-\sum_{v\in S(F_j^{\pi})}b_v}.
\]

Third, for the spanning-tree factor, fix \(T\in\mathcal T\). Since \(x_h\le 1\) for \(h\in E\setminus F\), we have
\[
\prod_{g\in T}x_g
\le
r^{|T\cap F|}
\prod_{j=1}^{m-1} t_j^{\,|T\cap F_j^{\pi}|}.
\]
Indeed,
\[
\prod_{g\in T} x_g
=
\Bigl(\prod_{g\in T\cap F} x_g\Bigr)
\Bigl(\prod_{h\in T\setminus F} x_h\Bigr).
\]
Since \(x_h\le 1\) for every \(h\in E\setminus F\), we have \(\prod_{g\in T} x_g \le \prod_{g\in T\cap F} x_g\). Now write \(T\cap F=\{g_{k_1},\dots,g_{k_q}\}\), where \(q=|T\cap F|\). For each such edge \(g_k\), we have \(x_{g_k}=r\prod_{\ell=k}^{m-1} t_\ell\). Therefore
\[
\prod_{g\in T\cap F} x_g
=
\prod_{g_k\in T\cap F}\left(r\prod_{\ell=k}^{m-1} t_\ell\right)
=
r^{|T\cap F|}
\prod_{j=1}^{m-1} t_j^{\,\#\{g_k\in T\cap F:\ k\le j\}}.
\]
But \(\#\{g_k\in T\cap F:\ k\le j\}=|T\cap F_j^{\pi}|\), since \(F_j^{\pi}=\{g_1,\dots,g_j\}\). Hence
\[
\prod_{g\in T}x_g
\le
r^{|T\cap F|}
\prod_{j=1}^{m-1} t_j^{\,|T\cap F_j^{\pi}|}.
\]

Because a spanning tree must reconnect the \(\kappa(G\setminus F_j^{\pi})\) components of \(G\setminus F_j^{\pi}\), one has \(|T\cap F_j^{\pi}|\ge \kappa(G\setminus F_j^{\pi})-1\). This can be seen more directly by comparing \(G\setminus F_j^{\pi}\) with the spanning subgraph \(T\setminus F_j^{\pi}\): since \(T\setminus F_j^{\pi} \subseteq G\setminus F_j^{\pi}\) and both graphs have the same vertex set, adding edges can only decrease the number of connected components, so \(\kappa(G\setminus F_j^{\pi})\le \kappa(T\setminus F_j^{\pi})\). Now \(T\) is a tree, so removing \(k\) edges from \(T\) produces exactly \(k+1\) connected components; taking \(k=|T\cap F_j^{\pi}|\) gives \(\kappa(T\setminus F_j^{\pi})=|T\cap F_j^{\pi}|+1\). Combining these yields
\[
\kappa(G\setminus F_j^{\pi})\le |T\cap F_j^{\pi}|+1.
\]

Hence
\[
\sqrt{\tau(x)}
\le
\sqrt{|\mathcal T|}\,
r^{(\kappa(G\setminus F)-1)/2}
\prod_{j=1}^{m-1} t_j^{\,(\kappa(G\setminus F_j^{\pi})-1)/2}.
\]

Also, by Lemma~\ref{lem:ordered_small_edge_coordinates},
\[
x_e^{-s}=r^{-s}\prod_{j=1}^{m-1} t_j^{-s\,\1(e\in F_j^{\pi})}.
\]

Recall that on \(\Omega(F,\pi)\),
\[
x_e^{-s}\rho_{v_0,A}(x)\,dx
=
\widetilde Z_{v_0,A}^{-1}\,
x_e^{-s}
\Bigl(\prod_{g\in F}x_g^{a_g-1}\Bigr)
\Bigl(\prod_{h\in E\setminus F}x_h^{a_h-1}\Bigr)
\Bigl(\prod_{v\in V}x_v^{-b_v}\Bigr)
\sqrt{\tau(x)}\,dx_F\,dz.
\]
Using the preceding bounds together with Lemma~\ref{lem:ordered_small_edge_coordinates},
\[
\Bigl(\prod_{g\in F}x_g^{a_g-1}\Bigr)\,dx_F
=
r^{a(F)-1}\prod_{j=1}^{m-1} t_j^{a(F_j^{\pi})-1}\,dr\,dt_1\cdots dt_{m-1},
\qquad
x_e^{-s}
=
r^{-s}\prod_{j=1}^{m-1} t_j^{-s\,\1(e\in F_j^{\pi})}.
\]
Hence
\[
x_e^{-s}\rho_{v_0,A}(x)\,dx
\le
\frac{\sqrt{|\mathcal T|}\,M_{\mathrm{edge}}(F)\,\delta_0^{-B_{\mathrm{out}}(F)}}{\widetilde Z_{v_0,A}}
\,r^{A(F)-1-s}
\prod_{j=1}^{m-1} t_j^{A(F_j^{\pi})-1-s\,\1(e\in F_j^{\pi})}
\,dr\,dt_1\cdots dt_{m-1}\,dz,
\]
because the exponent of \(r\) is \(a(F)-1-\sum_{v\in S(F)}b_v+(\kappa(G\setminus F)-1)/2-s=A(F)-1-s\), and similarly the exponent of \(t_j\) is \(a(F_j^{\pi})-1-\sum_{v\in S(F_j^{\pi})}b_v+(\kappa(G\setminus F_j^{\pi})-1)/2-s\,\1(e\in F_j^{\pi})=A(F_j^{\pi})-1-s\,\1(e\in F_j^{\pi})\). Therefore
\[
x_e^{-s}\rho_{v_0,A}(x)\,dx
\le
\frac{\sqrt{|\mathcal T|}\,M_{\mathrm{edge}}(F)\,\delta_0^{-B_{\mathrm{out}}(F)}}{\widetilde Z_{v_0,A}}
\,r^{A(F)-1-s}
\prod_{j=1}^{m-1} t_j^{A(F_j^{\pi})-1-s\,\1(e\in F_j^{\pi})}
\,dr\,dt_1\cdots dt_{m-1}\,dz.
\]
where \(dz\) denotes the remaining simplex coordinates outside \(F\). To make the \(z\)-coordinates precise, fix once and for all an edge \(h_\ast\in E\setminus F\) (this is possible since \(F\subsetneq E\)). Write
\[
E\setminus F=\{h_\ast\}\cup H,
\qquad |H|=|E|-|F|-1.
\]
For \(h\in H\), set \(z_h:=x_h\), and let \(dz:=\prod_{h\in H}dz_h\). Then, on \(\Delta_E\), the remaining outside coordinate is determined by the simplex constraint:
\[
x_{h_\ast}
=
1-\sum_{g\in F}x_g-\sum_{h\in H} z_h.
\]
Thus \((r,t_1,\dots,t_{m-1},z)\) gives a coordinate system on \(\Omega(F,\pi)\), where \(z=(z_h)_{h\in H}\in \mathbb R^{|E|-|F|-1}\). More explicitly, the admissible \(z\)-set is
\[
D(r,t):=
\Bigl\{
(z_h)_{h\in H}:\ z_h\ge \delta_0\ \forall h\in H,\ \ 
x_{h_\ast}=1-\sum_{g\in F}x_g-\sum_{h\in H}z_h\ge \delta_0
\Bigr\}.
\]
In particular, since each coordinate satisfies \(\delta_0\le z_h\le 1\), we have
\[
D(r,t)\subseteq [\delta_0,1]^{|E|-|F|-1}\subseteq [0,1]^{|E|-|F|-1}.
\]
Hence its Lebesgue measure is at most \(1\):
\[
|D(r,t)|\le 1.
\]
Therefore, after obtaining an upper bound that is independent of \(z\), integrating over the \(z\)-variables contributes at most a factor \(1\). If \(|E\setminus F|=1\), then \(H=\varnothing\), there are no \(z\)-variables, and this factor is interpreted as \(1\).

By Lemma \ref{lem:A-lower-bound-explicit}, every exponent \(A(F)\) and \(A(F_j^{\pi})\) is at least \(\underline a/2\). Since \(0<s<\underline a/2\), all the resulting integrals are finite, and
\[
\int_0^{\delta_0} r^{A(F)-1-s}\,dr
\le
\frac{\delta_0^{\,\underline a/2-s}}{\underline a/2-s},
\qquad
\int_0^1 t_j^{A(F_j^{\pi})-1-s\,\1(e\in F_j^{\pi})}\,dt_j
\le
\frac{1}{\underline a/2-s}.
\]
Hence
\[
\int_{\Omega(F,\pi)} x_e^{-s}\rho_{v_0,A}(x)\,dx
\le
\frac{\sqrt{|\mathcal T|}\,M_{\mathrm{edge}}(F)\,\delta_0^{-B_{\mathrm{out}}(F)}}{\widetilde Z_{v_0,A}}
\frac{\delta_0^{\,\underline a/2-s}}{(\underline a/2-s)^{|F|}}.
\]
Summing over the \(|F|!\) possible orderings \(\pi\) and then over all nonempty proper \(F\subsetneq E\) containing \(e\) gives \eqref{eq:small_edge_moment_bound}.

Finally, since \(\E[X_e^{-s}]<\infty\), Markov's inequality gives
\[
\P(X_e<\varepsilon)=\P(X_e^{-s}>\varepsilon^{-s})
\le \varepsilon^s\,\E[X_e^{-s}],
\]
and substituting \eqref{eq:small_edge_moment_bound} yields \eqref{eq:small_edge_tail_bound}.
\end{proof}

The preceding theorem has two immediate consequences that will be used later.

\begin{lemma}[Root-mass tail and Laplace consequence]
\label{lem:root_tail_laplace}
Assume that \(a_g\ge \underline a>0\) for all \(g\in E\), and fix \(0<s<\underline a/2\).

\begin{enumerate}
    \item Let \(\pi^W\) be the stationary distribution associated with the random environment \(W\). 
    Then, for every \(\varepsilon>0\),
    \begin{equation}
    \label{eq:root_tail_bound}
    \P\bigl(\pi^W(v_0)<\varepsilon\bigr)
    \le
    2^s \Big(\min_{g:v_0\in g}C_{g,s}\Big) \varepsilon^s.
    \end{equation}

    \item For every \(u>0\) and every edge \(e\in E\),
    \begin{equation}
    \label{eq:laplace_Xe2_bound}
    \E[e^{-uX_e^2}]
    \le
    C_{s/2}\,u^{-s/2}\,\E[X_e^{-s}],
    \qquad
    C_{s/2}:=\sup_{y>0} y^{s/2}e^{-y}
    =\Bigl(\frac{s}{2e}\Bigr)^{s/2}.
    \end{equation}
    In particular, \(\E[e^{-uX_e^2}]\lesssim u^{-s/2}\), with an explicit constant depending only on \(e\) and \(s\).
\end{enumerate}
\end{lemma}

\begin{proof}
Since
\[
\pi^W(v_0)
=
\frac{W_{v_0}}{\sum_{u\in V}W_u}
=
\frac12\sum_{g\ni v_0}X_g,
\]
the event \(\{\pi^W(v_0)<\varepsilon\}\) implies
\(
\sum_{g\ni v_0}X_g<2\varepsilon.
\)
In particular, for every edge \(g\) incident to \(v_0\), one has
\(X_g<2\varepsilon\). Fix one such edge \(g_0\ni v_0\) that minimizes \(\{C_{g,s}:\,g\ni v_0\}\). Then
\[
\{\pi^W(v_0)<\varepsilon\}
\subseteq
\{X_{g_0}<2\varepsilon\}.
\]
Therefore, by \Cref{thm:small_edge_tail},
\[
\P(\pi^W(v_0)<\varepsilon)
\le
\P(X_{g_0}<2\varepsilon)
\le
2^s C_{g_0,s}\varepsilon^s.
\]

For the Laplace estimate, let \(Z>0\) and \(r>0\). Since \(e^{-y}\le C_r y^{-r}\) for all \(y>0\), where \(C_r:=\sup_{y>0}y^r e^{-y}=(r/e)^r\), one has \(e^{-uZ}\le C_r u^{-r}Z^{-r}\). Taking \(Z=X_e^2\) and \(r=s/2\) gives \eqref{eq:laplace_Xe2_bound}.
\end{proof}


\paragraph{Step 2: quenched concentration.}

We now pass from tail bounds on the random environment to concentration estimates for the quenched walk. For a fixed environment \(w\), let \(P^w\) be the quenched transition matrix, let \(\pi^w\) be its stationary distribution, and recall that
\(X_e(w):=w_e/\sum_{g\in E}w_g\) for \(e\in E\). We also write
\[
m(w):=\min_{e\in E}X_e(w).
\]
For the spectral gap, let \(\lambda_2(P^w)\) denote the second largest eigenvalue of \(P^w\), let
\begin{equation}        \label{eq:gamma_for_gap}
\gamma(P^w):=1-\lambda_2(P^w).
\end{equation}

\begin{proposition}[Spectral gap and edge-count concentration]
\label{prop:quenched_concentration_package}
For every fixed environment \(w\), the quenched chain \(P^w\) satisfies
\begin{equation}
\label{eq:spectral_gap_from_min_edge_main}
\gamma(P^w)\ge \frac{m(w)^2}{2}.
\end{equation}
Consequently, for every \(0<s<\underline a/2\),
\begin{equation}
\label{eq:spectral_gap_tail_main}
\Perrw^{v_0, A}\bigl(\gamma(P^W)<\varepsilon\bigr)
\le
\Bigl(\sum_{e\in E} C_{e,s}\Bigr)\,(2\varepsilon)^{s/2},
\qquad 0<\varepsilon<\frac12,
\end{equation}
where \(W\sim\mu_{v_0,A}\) and the constants \(C_{e,s}\) are those from \Cref{thm:small_edge_tail}.

Now fix an edge \(e\in E\), and let
\[
N_e(T):=\sum_{t=0}^{T-1}\1\bigl(\{X_t,X_{t+1}\}=e\bigr)
\]
be its 
traversal count up to time \(T\). 
If the chain starts from stationarity \(\pi^w\), then for every \(\eta>0\) and every \(T\ge 1\),
\begin{equation}
\label{eq:quenched_concentration_main}
\Psrw^{\pi^w,w}\!\left(\bigl|N_e(T)-TX_e(w)\bigr|\ge \eta T\right)
\le
2\exp\!\left(-\frac{\gamma(P^w)}{4}\eta^2T\right)
+
4\exp\!\left(-\frac{1}{8}\eta^2T\right).
\end{equation}
In particular, on the event \(\{\gamma(P^w)\ge \varepsilon_1,\ X_e(w)\ge \varepsilon_2\}\), one has for every \(0<\delta\le 1\),
\begin{equation}
\label{eq:quenched_relative_concentration_main}
\Psrw^{\pi^w,w}\!\left(\bigl|N_e(T)-TX_e(w)\bigr|\ge \delta\,T X_e(w)\right)
\le
2\exp\!\left(-\frac{\varepsilon_1\varepsilon_2^2}{4}\delta^2T\right)
+
4\exp\!\left(-\frac{\varepsilon_2^2}{8}\delta^2T\right).
\end{equation}
\end{proposition}

\begin{proof}
For the spectral gap bound, let \(Q^w(i,j):=\pi^w(i)P^w_{ij}\). Since \(\pi^w(i)=w_i/\sum_{u\in V}w_u\) and \(P^w_{ij}=w_{ij}/w_i\), one has
\[
Q^w(i,j)=\frac{w_{ij}}{\sum_{u\in V}w_u}=\frac{X_{\{i,j\}}(w)}{2}.
\]
Hence, for every nonempty proper subset \(S\subsetneq V\),
\[
Q^w(S,S^c):=\sum_{i\in S,\ j\in S^c}Q^w(i,j)=\frac12\sum_{e\in \partial S}X_e(w).
\]
Here \(Q^w(S,S^c)\) is the total stationary flow from \(S\) to its complement \(S^c\), where \(Q^w(i,j):=\pi_i^wP_{ij}^w\), and \(\partial S:=\{\{i,j\}\in E:\ i\in S,\ j\in S^c\}\) is the edge boundary of \(S\). Because \(G\) is connected, \(\partial S\neq\varnothing\), so \(Q^w(S,S^c)\ge m(w)/2\). Therefore the conductance of the chain satisfies
\[
\Phi(P^w):=
\min_{\substack{S\subset V\\ 0<\pi^w(S)\le 1/2}}
\frac{Q^w(S,S^c)}{\pi^w(S)}
\ge
\min_{\emptyset \neq S\subsetneq V}
2Q^w(S,S^c)
=
\min_{\emptyset \neq S\subsetneq V}
\sum_{e\in \partial S}X_e(w)
\ge m(w).
\]
Cheeger's inequality for finite reversible Markov chains 
states that
\[
\gamma(P^w)\ge \frac{\Phi(P^w)^2}{2};
\]
see, for example, \cite[Theorem~13.14]{LevinPeresWilmer2017}.
Together with the inequality $\Phi(P^w) \ge m(w)$, we obtain \eqref{eq:spectral_gap_from_min_edge_main}.

To obtain \eqref{eq:spectral_gap_tail_main}, note that
\[
\{\gamma(P^W)<\varepsilon\}
\subseteq
\{m(W)<\sqrt{2\varepsilon}\}
=
\bigcup_{e\in E}\{X_e(W)<\sqrt{2\varepsilon}\}.
\]
Hence, by the union bound and \Cref{thm:small_edge_tail},
\[
\P\bigl(\gamma(P^W)<\varepsilon\bigr)
\le
\sum_{e\in E}\P\bigl(X_e(W)<\sqrt{2\varepsilon}\bigr)
\le
\Bigl(\sum_{e\in E}C_{e,s}\Bigr)(2\varepsilon)^{s/2}.
\]

We now turn to the concentration of \(N_e(T)\) for fixed \(w\). Define
\[
Y_t:=\1\bigl(\{X_t,X_{t+1}\}=e\bigr),
\qquad
f_e(i):=\Psrw^{v_0, w}\bigl(\{X_t,X_{t+1}\}=e\mid X_t=i\bigr)
=\sum_{j:\{i,j\}=e}P^w_{ij}.
\]
Then \(0\le f_e(i)\le 1\), and, under stationarity,
\[
\pi^w(f_e)
:=
\sum_{i\in V}\pi^w(i)f_e(i)
=
\Psrw^{\pi^w,w}\bigl(\{X_t,X_{t+1}\}=e\bigr)
=
X_e(w).
\]
Therefore \(\Esrw^{\pi^w,w}[N_e(T)]=TX_e(w)\).

Write
\[
A_T:=\sum_{t=0}^{T-1}\bigl(f_e(X_t)-X_e(w)\bigr)
\qquad\text{and}\qquad
M_T:=\sum_{t=0}^{T-1}\bigl(Y_t-f_e(X_t)\bigr).
\]
Then \(N_e(T)-TX_e(w)=A_T+M_T\).

To bound \(A_T\), we apply Theorem~1 of L\'eon--Perron \cite{LeonPerron2004}
to the stationary reversible chain \((X_t)\) under \(\Psrw^{\pi^w,w}\) and the function \(f_e\colon V\to[0,1]\). Since
\[
\Esrw^{\pi^w,w}[f_e(X_0)]=X_e(w),
\]
and, by stationarity, \(\sum_{t=0}^{T-1}f_e(X_t)\) has the same law as \(\sum_{t=1}^{T}f_e(X_t)\), Theorem~1 
of \cite{LeonPerron2004},
with \(n=T\), \(\mu=X_e(w)\), and \(\varepsilon=\eta/2\) yields, 
\[
\Psrw^{\pi^w,w}\!\left(A_T\ge \frac{\eta T}{2}\right)
=
\Psrw^{\pi^w,w}\!\left(\sum_{t=0}^{T-1}f_e(X_t)\ge T\Bigl(X_e(w)+\frac{\eta}{2}\Bigr)\right)
\le
\exp\!\left(
-\min\{1, \gamma(P^w)\}\,T\Bigl(\frac{\eta}{2}\Bigr)^2
\right).
\]

To control the lower tail, set \(g_e:=1-f_e\in[0,1]\). Then \(\E_{\pi^w,w}[g_e(X_0)]=1-X_e(w)\), and
\[
\left\{A_T\le -\frac{\eta T}{2}\right\}
=
\left\{\sum_{t=0}^{T-1}f_e(X_t)\le T\Bigl(X_e(w)-\frac{\eta}{2}\Bigr)\right\}
=
\left\{\sum_{t=0}^{T-1}g_e(X_t)\ge T\Bigl(1-X_e(w)+\frac{\eta}{2}\Bigr)\right\}.
\]
Therefore, we can apply the previous upper tail inequality to \(g_e\), 
getting
\[
\Psrw^{\pi^w,w}\!\left(A_T\le -\frac{\eta T}{2}\right)
\le
\exp\!\left(
-\min\{1,\gamma(P^w)\}\,T\Bigl(\frac{\eta}{2}\Bigr)^2
\right).
\]
Therefore, by the union bound,
\[
\Psrw^{\pi^w,w}\!\left(|A_T|\ge \frac{\eta T}{2}\right)
\le
2\exp\!\left(
-\frac14 \eta^2 T
\right)
+
2\exp\!\left(
-\frac14 \eta^2 \gamma(P^w)
\right).
\]

For \(M_T\), note that it is a martingale with respect to \(\mathcal F_t:=\sigma(X_0,\dots,X_t)\), because \(\E_w[Y_t\mid \mathcal F_t]=f_e(X_t)\). Its increments satisfy \(|M_{t+1}-M_t|=|Y_t-f_e(X_t)|\le 1\), so 
the
Azuma--Hoeffding 
inequality
\cite{Azuma1967,Hoeffding1963}
gives
\[
\Psrw^{\pi^w,w}\!\left(|M_T|\ge \frac{\eta T}{2}\right)
\le
2\exp\!\left(-\frac{1}{8}\eta^2T\right).
\]
Since
\[
\{|N_e(T)-TX_e(w)|\ge \eta T\}
\subseteq
\{|A_T|\ge \eta T/2\}\cup \{|M_T|\ge \eta T/2\},
\]
we obtain \eqref{eq:quenched_concentration_main}.

Finally, substituting \(\eta=\delta X_e(w)\) into \eqref{eq:quenched_concentration_main} and using \(\gamma(P^w)\ge \varepsilon_1, X_e(w)\ge \varepsilon_2\) yields \eqref{eq:quenched_relative_concentration_main}.
\end{proof}

To pass from 
a
stationary start to 
a start from 
the deterministic initial state \(v_0\), we use the following standard comparison estimate. The following lemma is inspired by~\cite[Proposition~3.15]{Paulin2015}.

\begin{lemma}[Transfer from 
a
stationary start to 
a start from a 
deterministic root 
]
\label{lem:nonstationary_transfer_v0}
Fix an environment \(w\), and let \(\pi^w\) be the a stationary distribution of the quenched chain \(P^w\). Then, for every event \(A\) measurable with respect to \((X_0,\dots,X_T)\),
\begin{equation}
\label{eq:root_transfer_main}
\Psrw^{v_0,w}(A)\le (\pi^w(v_0))^{-1/2}\,\Psrw^{\pi^w,w}(A)^{1/2}.
\end{equation}
In particular, on the event \(\{\pi^w(v_0)\ge \varepsilon_0\}\), one has
\begin{equation}
\label{eq:root_transfer_threshold_main}
\Psrw^{v_0,w}(A)\le \varepsilon_0^{-1/2}\,\Psrw^{\pi^w,w}(A)^{1/2}.
\end{equation}
\end{lemma}

\begin{proof}
We argue directly on path space. Let \(\mathcal X_T:=V^{T+1}\), and for \(x_0,\dots,x_T\in V\), note
\[
\Psrw^{v_0,w}(X_0^T=x_0^T)
=
\1(x_0=v_0)\prod_{t=0}^{T-1}p^w(x_t,x_{t+1})
\]
for the quenched path law started from \(v_0\), and
\[
\Psrw^{\pi^w,w}(X_0^T=x_0^T)
:=
\pi^w(x_0)\prod_{t=0}^{T-1}p^w(x_t,x_{t+1})
\]
for the quenched path law started from stationarity \(\pi^w\). Since \(\pi^w(x)>0\) for every \(x\in V\), the first measure is absolutely continuous with respect to the second, and
\[
\frac{\Psrw^{v_0,w}(X_0^T=x_0^T)}{\Psrw^{\pi^w,w}(X_0^T=x_0^T)}
=
\frac{\1(x_0=v_0)}{\pi^w(v_0)}.
\]

Therefore, for every event \(A\) measurable with respect to \((X_0,\dots,X_T)\),
\[
\Psrw^{v_0,w}(A)
=
\Esrw^{\pi^w,w}\!\left[\frac{\1(X_0=v_0)}{\pi^w(v_0)}\,\1_A\right].
\]
Applying 
the 
Cauchy--Schwarz
inequality, we have
\[
\Psrw^{v_0,w}(A)
\le
\left(
\Esrw^{\pi^w,w}\!\left[\left(\frac{\1(X_0=v_0)}{\pi^w(v_0)}\right)^2\right]
\right)^{1/2}
\Bigl(\Esrw^{\pi^w,w}[\1_A]\Bigr)^{1/2}.
\]
Now under \(\P_{\pi^w,w}\), one has \(X_0\sim \pi^w\), so
\[
\Esrw^{\pi^w,w}\!\left[\left(\frac{\1(X_0=v_0)}{\pi^w(v_0)}\right)^2\right]
=
\sum_{x\in V}\pi^w(x)\left(\frac{\1(x=v_0)}{\pi^w(v_0)}\right)^2
=
\frac{1}{\pi^w(v_0)},
\]
while \(\Esrw^{\pi^w,w}[\1_A]=\Esrw^{\pi^w,w}(A)\). Hence
\[
\Psrw^{v_0,w}(A)\le \pi^w(v_0)^{-1/2}\,\Psrw^{\pi^w,w}(A)^{1/2},
\]
which is \eqref{eq:root_transfer_main}. The bound \eqref{eq:root_transfer_threshold_main} follows immediately on the event \(\{\pi^w(v_0)\ge \varepsilon_0\}\).
\end{proof}


\paragraph{Step 3: thresholded and refined inverse-local-time bounds.}

We now combine the environment tail bounds from Step 1 with the quenched concentration estimates from Step 2. The first estimate is obtained by imposing an explicit lower threshold on \(X_e\); the second treats \(X_e\) directly through its negative moments and Laplace transform.

\begin{theorem}[Thresholded and refined inverse-local-time bounds]
\label{thm:inverse_local_time_master}
Assume that \(G=(V,E)\) is finite and connected, that the ERRW starts from \(v_0\), and that the initial edge weights satisfy \(0<\underline a\le a_g\le \overline a\) for all \(g\in E\). Fix an edge \(e\in E\) and \(\gamma>0\). Let \(W\sim \mu_{v_0,A}\), let \(P^W\) be the quenched chain, let \(\pi^W\) be its stationary distribution, and write \(X_e:=W_e/\sum_{g\in E}W_g\). Also let
\[
N_e(T):=\sum_{t=0}^{T-1}\1\bigl(\{X_t,X_{t+1}\}=e\bigr)
\]
be the traversal count of \(e\) up to time \(T\).

\begin{enumerate}
    \item[(i)] Let \(0<s<\underline a/2\), and define
    \[
    C_{\mathrm{gap},s}:=\sum_{h\in E}C_{h,s},
    \qquad
    C_{v_0,s}:=2^s\min_{g\ni v_0}C_{g,s},
    \]
    where the constants \(C_{h,s}\) are those from \Cref{thm:small_edge_tail}. Then, for every \(0<\varepsilon_0<1/2\), \(0<\varepsilon_1\le 1/2\), \(0<\varepsilon_2<1\), and every \(T\ge 1\),
    \begin{align}
    \Eerrw^{v_0, A}\!\left[\frac{1}{N_e(T)+\gamma}\right]
    &\le
    \frac{2}{T\varepsilon_2}
    +
    \frac{1}{\gamma}
    \Bigl[
    C_{e,s}\,\varepsilon_2^s
    +
    C_{\mathrm{gap},s}\,(2\varepsilon_1)^{s/2}
    +
    C_{v_0,s}\,\varepsilon_0^s
    \Bigr]
    \notag\\
    &\qquad
    +
    \frac{2}{\gamma}\,\varepsilon_0^{-1/2}
    \left(
    e^{-\varepsilon_1\varepsilon_2^2T/32}
    +
    e^{-\varepsilon_2^2T/64}
    \right).
    \label{eq:inverse_local_time_thresholded}
    \end{align}

    \item[(ii)] If in addition \(1<s<\underline a/2\), then for every \(0<\varepsilon_0<1/2\), \(0<\varepsilon_1\le 1/2\), and every \(T\ge 1\),
    \begin{align}
    \Eerrw^{v_0,A}\!\left[\frac{1}{N_e(T)+\gamma}\right]
    &\le
    2\,C_{e,s}\,T^{-1}
    +
    \frac{2^{s/2}}{\gamma}\,C_{\mathrm{gap},s}\,\varepsilon_1^{s/2}
    +
    \frac{1}{\gamma}\,C_{v_0,s}\,\varepsilon_0^s
    \notag\\
    &\qquad
    +
    \frac{2\cdot 32^{s/2}\,C_{s/2}\,C_{e,s}}{\gamma}\,
    \varepsilon_0^{-1/2}(\varepsilon_1T)^{-s/2}
    +
    \frac{2\cdot 64^{s/2}\,C_{s/2}\,C_{e,s}}{\gamma}\,
    \varepsilon_0^{-1/2}T^{-s/2}.
    \label{eq:inverse_local_time_refined}
    \end{align}
\end{enumerate}
\end{theorem}

\begin{proof}
We prove (i) and (ii) separately.

\smallskip
\noindent\emph{Proof of (i).}
Let
\[
B:=
\{X_e<\varepsilon_2\}
\cup
\{\gamma(P^W)<\varepsilon_1\}
\cup
\{\pi^W(v_0)<\varepsilon_0\},
\]
where we recall that $\gamma(P^W)$ is defined in \eqref{eq:gamma_for_gap}.
Splitting according to \(B\), we write
\[
 \Eerrw^{v_0, A}\!\left[\frac{1}{N_e(T)+\gamma}\right]
=
 \Eerrw^{v_0, A}\!\left[\frac{1}{N_e(T)+\gamma}\,\1_B\right]
+
 \Eerrw^{v_0, A}\!\left[\frac{1}{N_e(T)+\gamma}\,\1_{B^c}\right].
\]

On \(B\), the trivial bound \((N_e(T)+\gamma)^{-1}\le \gamma^{-1}\) gives
\[
 \Eerrw^{v_0, A}\!\left[\frac{1}{N_e(T)+\gamma}\,\1_B\right]
\le
\frac{1}{\gamma}\,\P(B).
\]
By the union bound together with \Cref{thm:small_edge_tail}, \eqref{eq:spectral_gap_tail_main}, and \eqref{eq:root_tail_bound}, we have
\[
\P(B)
\le
C_{e,s}\,\varepsilon_2^s
+
C_{\mathrm{gap},s}\,(2\varepsilon_1)^{s/2}
+
C_{v_0,s}\,\varepsilon_0^s.
\]
Hence
\[
 \Eerrw^{v_0, A}\!\left[\frac{1}{N_e(T)+\gamma}\,\1_B\right]
\le
\frac{1}{\gamma}
\Bigl[
C_{e,s}\,\varepsilon_2^s
+
C_{\mathrm{gap},s}\,(2\varepsilon_1)^{s/2}
+
C_{v_0,s}\,\varepsilon_0^s
\Bigr].
\]

We next estimate the contribution on \(B^c\). Fix \(w\in B^c\). Then \(X_e(w)\ge \varepsilon_2\), \(\gamma(P^w)\ge \varepsilon_1\), and \(\pi^w(v_0)\ge \varepsilon_0\). Let
\[
A_w:=\left\{\bigl|N_e(T)-TX_e(w)\bigr|\ge \frac12\,T X_e(w)\right\}.
\]
Applying \eqref{eq:quenched_relative_concentration_main} with \(\delta=1/2\), we obtain
\[
\Psrw^{\pi^w,w}(A_w)
\le
2\exp\!\left(-\varepsilon_1\varepsilon_2^2T/16\right)
+
4\exp\!\left(-\varepsilon_2^2T/32\right).
\]
By Lemma \ref{lem:nonstationary_transfer_v0},
\[
\Psrw^{v_0,w}(A_w)
\le
\pi^w(v_0)^{-1/2}\,\Psrw^{\pi^w,w}(A_w)^{1/2}
\le
2\,\varepsilon_0^{-1/2}
\left(
e^{-\varepsilon_1\varepsilon_2^2T/32}
+
e^{-\varepsilon_2^2T/64}
\right).
\]

On \(A_w^c\), one has \(N_e(T)\ge \frac12\,TX_e(w)\ge \frac12\,T\varepsilon_2\), so \((N_e(T)+\gamma)^{-1}\le 2/(T\varepsilon_2)\). On \(A_w\), we again use \((N_e(T)+\gamma)^{-1}\le \gamma^{-1}\). Thus
\[
 \Esrw^{v_0, w}\!\left[\frac{1}{N_e(T)+\gamma}\right]
\le
\frac{2}{T\varepsilon_2}
+
\frac{1}{\gamma}\,\Psrw^{v_0,w}(A_w).
\]
Averaging this bound over all \(w\in B^c\) gives
\[
 \Eerrw^{v_0, A}\!\left[\frac{1}{N_e(T)+\gamma}\,\1_{B^c}\right]
\le
\frac{2}{T\varepsilon_2}
+
\frac{2}{\gamma}\,\varepsilon_0^{-1/2}
\left(
e^{-\varepsilon_1\varepsilon_2^2T/32}
+
e^{-\varepsilon_2^2T/64}
\right).
\]
Combining the estimates on \(B\) and \(B^c\) yields \eqref{eq:inverse_local_time_thresholded}.

\smallskip
\noindent\emph{Proof of (ii).}
Let
\[
B':=\{\gamma(P^W)<\varepsilon_1\}\cup \{\pi^W(v_0)<\varepsilon_0\}.
\]
Again we split
\[
 \Eerrw^{v_0, A}\!\left[\frac{1}{N_e(T)+\gamma}\right]
=
 \Eerrw^{v_0, A}\!\left[\frac{1}{N_e(T)+\gamma}\,\1_{B'}\right]
+
 \Eerrw^{v_0, A}\!\left[\frac{1}{N_e(T)+\gamma}\,\1_{(B')^c}\right].
\]

On \(B'\), we use \((N_e(T)+\gamma)^{-1}\le \gamma^{-1}\), so
\[
 \Eerrw^{v_0, A}\!\left[\frac{1}{N_e(T)+\gamma}\,\1_{B'}\right]
\le
\frac{1}{\gamma}\,\Perrw^{v_0, A}(B').
\]
By \eqref{eq:spectral_gap_tail_main} and \eqref{eq:root_tail_bound},
\[
\Perrw^{v_0, A}(B')
\le
C_{\mathrm{gap},s}\,(2\varepsilon_1)^{s/2}
+
C_{v_0,s}\,\varepsilon_0^s
=
2^{s/2}C_{\mathrm{gap},s}\,\varepsilon_1^{s/2}
+
C_{v_0,s}\,\varepsilon_0^s.
\]
Hence
\[
 \Eerrw^{v_0, A}\!\left[\frac{1}{N_e(T)+\gamma}\,\1_{B'}\right]
\le
\frac{2^{s/2}}{\gamma}C_{\mathrm{gap},s}\,\varepsilon_1^{s/2}
+
\frac{1}{\gamma}C_{v_0,s}\,\varepsilon_0^s.
\]

Fix now \(w\in (B')^c\). Then \(\gamma(P^w)\ge \varepsilon_1\) and \(\pi^w(v_0)\ge \varepsilon_0\). Define
\[
A_w:=\left\{\bigl|N_e(T)-TX_e(w)\bigr|\ge \frac12\,T X_e(w)\right\}.
\]
Applying \eqref{eq:quenched_concentration_main} with \(\eta=X_e(w)/2\), and using \(\gamma(P^w)\ge \varepsilon_1\), gives
\[
\Psrw^{\pi^w,w}(A_w)
\le
2\exp\!\left(-\varepsilon_1X_e(w)^2T/16\right)
+
4\exp\!\left(-X_e(w)^2T/32\right).
\]
By Lemma \ref{lem:nonstationary_transfer_v0},
\[
\Psrw^{v_0,w}(A_w)
\le
2\,\varepsilon_0^{-1/2}
\left(
e^{-\varepsilon_1X_e(w)^2T/32}
+
e^{-X_e(w)^2T/64}
\right).
\]

On \(A_w^c\), one has \(N_e(T)\ge \frac12\,TX_e(w)\), so \((N_e(T)+\gamma)^{-1}\le 2/(TX_e(w))\). Hence
\[
 \Esrw^{v_0, w}\!\left[\frac{1}{N_e(T)+\gamma}\right]
\le
\frac{2}{TX_e(w)}
+
\frac{1}{\gamma}\,\Psrw^{v_0,w}(A_w).
\]
Averaging over \(w\in (B')^c\), we obtain
\begin{align*}
 \Eerrw^{v_0, A}\!\left[\frac{1}{N_e(T)+\gamma}\,\1_{(B')^c}\right]
&\le
\frac{2}{T}\, \Eerrw^{v_0, A}[X_e^{-1}]
\\
&\qquad
+
\frac{2}{\gamma}\,\varepsilon_0^{-1/2}
\left(
 \Eerrw^{v_0, A}[e^{-\varepsilon_1X_e^2T/32}]
+
 \Eerrw^{v_0, A}[e^{-X_e^2T/64}]
\right).
\end{align*}

Since \(1<s<\underline a/2\) and \(0<X_e\le 1\), one has \(X_e^{-1}\le X_e^{-s}\), so
\[
\Eerrw^{v_0,A}[X_e^{-1}]
\le
\Eerrw^{v_0,A}[X_e^{-s}]
\le
C_{e,s}
\]
by \Cref{thm:small_edge_tail}. Also, by \eqref{eq:laplace_Xe2_bound},
\[
\Eerrw^{v_0,A}[e^{-\varepsilon_1X_e^2T/32}]
\le
C_{s/2}\,(\varepsilon_1T/32)^{-s/2}\,\Eerrw^{v_0,A}[X_e^{-s}]
\le
32^{s/2}C_{s/2}C_{e,s}\,(\varepsilon_1T)^{-s/2},
\]
and similarly
\[
\Eerrw^{v_0,A}[e^{-X_e^2T/64}]
\le
C_{s/2}\,(T/64)^{-s/2}\,\Eerrw^{v_0,A}[X_e^{-s}]
\le
64^{s/2}C_{s/2}C_{e,s}\,T^{-s/2}.
\]
Substituting these estimates gives
\begin{align*}
 \Eerrw^{v_0, A}\!\left[\frac{1}{N_e(T)+\gamma}\,\1_{(B')^c}\right]
&\le
2\,C_{e,s}\,T^{-1}
\\
&\qquad
+
\frac{2\cdot 32^{s/2}C_{s/2}C_{e,s}}{\gamma}\,
\varepsilon_0^{-1/2}(\varepsilon_1T)^{-s/2}
\\
&\qquad
+
\frac{2\cdot 64^{s/2}C_{s/2}C_{e,s}}{\gamma}\,
\varepsilon_0^{-1/2}T^{-s/2}.
\end{align*}
Combining this with the estimate on \(B'\) yields \eqref{eq:inverse_local_time_refined}.
\end{proof}

\paragraph{Step 4: optimization.}

We now optimize the two estimates from \Cref{thm:inverse_local_time_master} to obtain the final polynomial decay rate.

\begin{theorem}[Inverse-local-time decay]
\label{thm:inverse_local_time_simple}
Assume that \(G=(V,E)\) is finite and connected, that the ERRW starts from \(v_0\), and that the initial edge weights satisfy \(0<\underline a\le a_g\le \overline a\) for all \(g\in E\). Fix an edge \(e\in E\) and \(\gamma>0\).

\begin{enumerate}
    \item[(i)] If \(0<\underline a\le 4+2\sqrt5\), then for every
    \[
    0<r<\min\!\left\{\frac{\underline a}{8},\frac34\right\},
    \]
    there exists a constant \(C_r<\infty\), depending only on \(G,v_0,A,e,\gamma,r\), such that
    \[
    \Eerrw^{v_0,A}\!\left[\frac{1}{N_e(T)+\gamma}\right]\le C_r\,T^{-r},
    \qquad T\ge 1.
    \]

    \item[(ii)] If \(\underline a>4+2\sqrt5\), then there exists a constant \(C<\infty\), depending only on \(G,v_0,A,e,\gamma\), such that
    \[
    \Eerrw^{v_0,A}\!\left[\frac{1}{N_e(T)+\gamma}\right]\le \frac{C}{T},
    \qquad T\ge 1.
    \]
\end{enumerate}
\end{theorem}

\begin{proof}
Set \(E_T:=\Eerrw^{v_0,A}\!\left[(N_e(T)+\gamma)^{-1}\right]\).

\smallskip
\noindent\emph{Proof of (i).}
Assume \(0<\underline a\le 4+2\sqrt5\), and fix any constant \(r\) such that
\[
0<r<\min\!\left\{\frac{\underline a}{8},\frac34\right\}.
\]
Choose any constant \(s\) so that \(4r<s<\underline a/2\). This is possible because \(r<\underline a/8\) implies \(4r<\underline a/2\).

Next choose any constant \(\kappa>0\) so that
\[
\kappa<\frac12,\qquad \kappa<\frac34-r,\qquad \kappa<\frac{2}{s}\Bigl(\frac{s}{4}-r\Bigr).
\]
This is possible because \(r<3/4\) and \(r<s/4\). With this choice, \(r<3/4-\kappa\) and \(r<s/4-s\kappa/2\).

Apply \Cref{thm:inverse_local_time_master}(i) with
\[
\varepsilon_0=T^{-1},\qquad \varepsilon_1=T^{-1/2+\kappa},\qquad \varepsilon_2=T^{-1/4}.
\]
Then \(\frac{2}{T\varepsilon_2}=2T^{-3/4}\), \(\varepsilon_2^s=T^{-s/4}\), \(\varepsilon_1^{s/2}=T^{-s/4+s\kappa/2}\), and \(\varepsilon_0^s=T^{-s}\). Moreover,
\[
\varepsilon_1\varepsilon_2^2T
=
T^{-1/2+\kappa}\cdot T^{-1/2}\cdot T
=
T^\kappa,
\qquad
\varepsilon_2^2T
=
T^{-1/2}\cdot T
=
T^{1/2},
\]
and \(\varepsilon_0^{-1/2}=T^{1/2}\). Hence \eqref{eq:inverse_local_time_thresholded} gives
\[
E_T
\le
C\Bigl(
T^{-3/4}
+
T^{-s/4}
+
T^{-s/4+s\kappa/2}
+
T^{-s}
+
T^{1/2}e^{-cT^\kappa}
+
T^{1/2}e^{-cT^{1/2}}
\Bigr)
\]
for constants \(C,c>0\) depending only on \(G,v_0,A,e,\gamma,r\).

We now compare the exponents with \(r\). By construction, \(3/4>r\), \(s/4>r\), \(s/4-s\kappa/2>r\), and \(s>r\), so each polynomial term on the right-hand side is \(O(T^{-r})\).

It remains to control the exponential terms. We use the elementary inequality that for every \(c>0\) and every \(\beta>0\),
\[
\sup_{u\ge 0} u^\beta e^{-cu}<\infty.
\]
Indeed, the function \(\phi(u):=u^\beta e^{-cu}\) is continuous on \([0,\infty)\), tends to \(0\) as \(u\to\infty\), and for \(u>0\) satisfies \(\phi'(u)=u^{\beta-1}e^{-cu}(\beta-cu)\). Hence its maximum is attained at \(u=\beta/c\), and therefore
\[
u^\beta e^{-cu}\le \left(\frac{\beta}{ce}\right)^\beta,\qquad u\ge 0.
\]

Applying this with \(u=T^\kappa\) and \(\beta=(r+\tfrac12)/\kappa\), we obtain
\[
T^{r+1/2}e^{-cT^\kappa}
=
(T^\kappa)^{(r+1/2)/\kappa}e^{-cT^\kappa}
\le
\left(\frac{(r+1/2)/\kappa}{ce}\right)^{(r+1/2)/\kappa}.
\]
Thus \(T^{1/2}e^{-cT^\kappa}\le C_{r,\kappa,c}\,T^{-r}\) for all \(T\ge 1\).

Similarly, applying the same inequality with \(u=T^{1/2}\) and \(\beta=2r+1\), we get
\[
T^{r+1/2}e^{-cT^{1/2}}
=
(T^{1/2})^{2r+1}e^{-cT^{1/2}}
\le
\left(\frac{2r+1}{ce}\right)^{2r+1},
\]
and hence \(T^{1/2}e^{-cT^{1/2}}\le C'_{r,c}\,T^{-r}\) for all \(T\ge 1\).

Therefore every term on the right-hand side is \(O(T^{-r})\), and hence
\[
E_T\le C_rT^{-r},
\qquad T\ge 1.
\]

\smallskip
\noindent\emph{Proof of (ii).}
Assume \(\underline a>4+2\sqrt5\). Choose \(s\) so that
\[
2+\sqrt5<s<\frac{\underline a}{2}.
\]
This is possible because \(\underline a>4+2\sqrt5\) implies \(\underline a/2>2+\sqrt5\).

Apply \Cref{thm:inverse_local_time_master}(ii) with
\[
\varepsilon_0=T^{-1/s},
\qquad
\varepsilon_1=T^{-2/s}.
\]
Then \(\varepsilon_1^{s/2}=(T^{-2/s})^{s/2}=T^{-1}\) and \(\varepsilon_0^s=(T^{-1/s})^s=T^{-1}\). Also,
\[
\varepsilon_0^{-1/2}(\varepsilon_1T)^{-s/2}
=
T^{1/(2s)}\bigl(T^{-2/s}T\bigr)^{-s/2}
=
T^{1/(2s)}\,T^{-(1-2/s)s/2}
=
T^{1/(2s)}\,T^{-s/2+1}
=
T^{-(\,s/2-1-1/(2s)\,)},
\]
and similarly
\[
\varepsilon_0^{-1/2}T^{-s/2}
=
T^{1/(2s)}T^{-s/2}
=
T^{-(\,s/2-1/(2s)\,)}.
\]

We now check that both exponents are strictly larger than \(1\)
in absolute value. 
First,
\[
\frac{s}{2}-1-\frac{1}{2s}>1
\iff
s-\frac1s>4
\iff
s^2-4s-1>0.
\]
The positive root of \(s^2-4s-1=0\) is \(2+\sqrt5\), so this holds because \(s>2+\sqrt5\). Second,
\[
\frac{s}{2}-\frac{1}{2s}>1
\iff
s-\frac1s>2
\iff
s^2-2s-1>0,
\]
whose positive root is \(1+\sqrt2\); this also holds because \(s>2+\sqrt5\).

Therefore both terms \(\varepsilon_0^{-1/2}(\varepsilon_1T)^{-s/2}\) and \(\varepsilon_0^{-1/2}T^{-s/2}\) are \(O(T^{-1-\delta})\) for some \(\delta>0\), and in particular are \(O(T^{-1})\). Since the first three terms in \eqref{eq:inverse_local_time_refined} are already \(O(T^{-1})\), it follows that every term on the right-hand side of \eqref{eq:inverse_local_time_refined} is \(O(T^{-1})\). Hence
\[
E_T\le \frac{C}{T},
\qquad T\ge 1,
\]
for a finite constant \(C\) depending only on \(G,v_0,A,e,\gamma\).
\end{proof}
\begin{corollary}[Decay of the trajectory-level KL gap]
\label{cor:gap_decay_general}
Let \(A^{(0)},A^{(1)}\in(0,\infty)^E\), and let
\[
\mu_s:=\mu_{v_0,A^{(s)}},
\qquad
P_s^{(T)}:=\mathcal L_{\Perrw^{v_0,A^{(s)}}}(X_0^T),
\qquad s\in\{0,1\}.
\]
Define
\[
\Gap_T:=\KL(\mu_0\|\mu_1)-\KL(P_0^{(T)}\|P_1^{(T)}),
\qquad
\delta_e:=a_e^{(1)}-a_e^{(0)},
\qquad
\underline a:=\min_{e\in E}\min\{a_e^{(0)},a_e^{(1)}\} >0.
\]
Then the following hold.

\begin{enumerate}
    \item[(i)] If \(0<\underline a\le 4+2\sqrt5\), then for every
    \[
    0<r<\min\!\left\{\frac{\underline a}{8},\frac34\right\},
    \]
    there exists a constant \(C_r<\infty\), depending only on \(G,v_0,A^{(0)},A^{(1)},r\), such that
    \[
    \Gap_T\le C_r\,T^{-r},
    \qquad T\ge 1.
    \]

    \item[(ii)] If \(\underline a>4+2\sqrt5\), then there exists a constant \(C<\infty\), depending only on \(G,v_0,A^{(0)},A^{(1)}\), such that
    \[
    \Gap_T\le \frac{C}{T},
    \qquad T\ge 1.
    \]
\end{enumerate}
\end{corollary}

\begin{proof}
By Proposition \ref{prop:gap_general_bounds},
\[
\Gap_T
\le
\frac12\Bigl(1+\frac1{\underline a}\Bigr)
\sum_{e\in E}\delta_e^2\,
\E_{P_0^{(T)}}\!\left[\frac{1}{N_e(T)+\underline a}\right].
\]
Here the expectation is taken under the ERRW law with initial weights \(A^{(0)}\). Since \(E\) is finite and \(A^{(0)}\in(0,\infty)^E\), we have
\[
0<\underline a\le a_e^{(0)}\le \overline a_0:=\max_{e\in E}a_e^{(0)}<\infty,
\qquad e\in E.
\]
Thus \Cref{thm:inverse_local_time_simple} applies to ERRW with initial weight vector \(A^{(0)}\) and \(\gamma=\underline a\).

Assume first that \(0<\underline a\le 4+2\sqrt5\), and fix
\[
0<r<\min\!\left\{\frac{\underline a}{8},\frac34\right\}.
\]
For each edge \(e\in E\), \Cref{thm:inverse_local_time_simple}(i) yields a finite constant \(C_{r,e}'\), depending only on \(G,v_0,A^{(0)},e,\underline a,r\), such that
\[
\E_{P_0^{(T)}}\!\left[\frac{1}{N_e(T)+\underline a}\right]
\le
C_{r,e}'\,T^{-r},
\qquad T\ge 1.
\]
Since \(E\) is finite, we may define
\[
C_r':=\max_{e\in E} C_{r,e}'<\infty.
\]
Then
\[
\E_{P_0^{(T)}}\!\left[\frac{1}{N_e(T)+\underline a}\right]
\le
C_r' T^{-r},
\qquad T\ge 1,
\]
uniformly over \(e\in E\). Substituting this into the preceding bound gives
\[
\Gap_T
\le
\frac12\Bigl(1+\frac1{\underline a}\Bigr)
\Bigl(\sum_{e\in E}\delta_e^2\Bigr)\,C_r' T^{-r},
\]
which proves part (i).

Now assume that \(\underline a>4+2\sqrt5\). For each edge \(e\in E\), \Cref{thm:inverse_local_time_simple}(ii) yields a finite constant \(C_e''\), depending only on \(G,v_0,A^{(0)},e,\underline a\), such that
\[
\E_{P_0^{(T)}}\!\left[\frac{1}{N_e(T)+\underline a}\right]
\le
C_e''\,T^{-1},
\qquad T\ge 1.
\]
Again, since \(E\) is finite, we may define
\[
C'':=\max_{e\in E} C_e''<\infty,
\]
so that
\[
\E_{P_0^{(T)}}\!\left[\frac{1}{N_e(T)+\underline a}\right]
\le
C''\,T^{-1},
\qquad T\ge 1,
\]
uniformly over \(e\in E\). Inserting this into the same bound from Proposition \ref{prop:gap_general_bounds} yields
\[
\Gap_T
\le
\frac12\Bigl(1+\frac1{\underline a}\Bigr)
\Bigl(\sum_{e\in E}\delta_e^2\Bigr)\,C''\,T^{-1},
\]
which proves part (ii).
\end{proof}

\section{Discussion and Open Problems}
\label{sec:discussion}

This paper studies several information-theoretic quantities for ERRW on finite graphs, including the entropy rate, the KL divergence between environment laws, and the KL divergence between finite-trajectory laws. The random-environment representation makes these quantities accessible in a rather explicit form, and in particular reduces the convergence of trajectory-level KL divergence toward environment-level KL divergence to the control of inverse local times.

A first natural open problem is to obtain a sharp characterization of this convergence in the general graph case. In the \(n\)-star, the P\'olya-urn structure yields the precise rates \(T^{-1}\), \((\log T)T^{-1}\), and \(T^{-a}\). By contrast, for general finite graphs, our argument gives polynomial upper bounds, but these are unlikely to be optimal  for $\underline{a}< 4+2\sqrt{5}$. It would be very interesting to identify the correct rate, and more broadly to understand how it depends on the graph geometry and the initial weights.

A second, and perhaps more important, direction is to use these information-theoretic results to derive quantitative upper and lower bounds for testing problems involving ERRWs. The most immediate example is two-point testing, where the trajectory-level KL divergence governs the Stein exponent. More generally, one may ask about identity testing, closeness testing, and related minimax testing problems. From this perspective, the quantities studied here should be viewed not only as intrinsic objects of interest, but also as tools for proving information-theoretic lower bounds and for designing statistically efficient tests.

\bibliographystyle{plain}
\bibliography{refs} 

\appendix
\section{Tightness of the mutual-information upper bound on \(I(W;X_0^T)\)}
\label{sec:mi_tightness_star}

In the special case of the \(n\)-star, the logarithmic upper bound from Lemma \ref{lem:entropy_anneal} is sharp up to the value of the constant. Let \(G\) be the \(n\)-star with center \(v_0\) and leaves \(\ell_1,\dots,\ell_n\), and assume first that the walk starts from the center \(v_0\). Write \(e_i:=\{v_0,\ell_i\}\), let \(a_i:=a_{e_i}\), set \(\alpha_i:=a_i/2\), and let \(\alpha_0:=\sum_{i=1}^n \alpha_i\). For the random environment \(W\), define
\[
P_i:=\frac{W_{e_i}}{\sum_{j=1}^n W_{e_j}},
\qquad i=1,\dots,n.
\]
Then
\[
I(W;X_0^T)=\frac{n-1}{2}\log T+O(1)
\qquad (T\to\infty),
\]
or equivalently,
\[
\frac{1}{T}I(W;X_0^T)
=
\frac{n-1}{2}\frac{\log T}{T}+O(T^{-1}).
\]

\begin{proof}
Let \(M:=\lceil T/2\rceil\). Since the graph is a star and the walk starts from the center, the trajectory is deterministic except when it departs from \(v_0\). Let \(Y_r\in\{1,\dots,n\}\) be the leaf chosen on the \(r\)-th departure from \(v_0\), for \(r=1,\dots,M\). Then \(X_0^T\) and \((Y_1,\dots,Y_M)\) determine each other, so
\[
I(W;X_0^T)=I(W;Y_1,\dots,Y_M).
\]

Moreover, the law of \((Y_1,\dots,Y_M)\) depends on \(W\) only through the vector \(P=(P_1,\dots,P_n)\), since \(P_i\) is the quenched probability of choosing leaf \(\ell_i\) upon departure from the center. Hence \(I(W;Y_1,\dots,Y_M)=I(P;Y_1,\dots,Y_M)\).

Now let \(K_i(m):=\sum_{r=1}^m \1\{Y_r=i\}\). On the \((m+1)\)-st departure from the center, the edge \(e_i\) has weight \(a_i+2K_i(m)\), because each previous excursion using \(e_i\) contributes two traversals of that edge. Therefore
\[
\Pr(Y_{m+1}=i\mid Y_1,\dots,Y_m)
=
\frac{a_i+2K_i(m)}{\sum_{j=1}^n a_j+2m}
=
\frac{\alpha_i+K_i(m)}{\alpha_0+m}.
\]
Thus \((Y_r)_{r\ge 1}\) is exactly a P\'olya urn sequence with de Finetti measure \(\Dir(\alpha_1,\dots,\alpha_n)\). Equivalently, conditional on \(P\sim\Dir(\alpha_1,\dots,\alpha_n)\),
the variables \(Y_1,\dots,Y_M\) are
independent and satisfy
\[
\Pr(Y_r=i\mid P)=P_i,\qquad i=1,\dots,n.
\]
In other words, conditional on \(P\), each \(Y_r\) has the categorical distribution
on \(\{1,\dots,n\}\) with probability vector \(P\). This is the standard
de Finetti representation of P\'olya's urn, or equivalently the usual
Dirichlet--multinomial mixture representation; see, for example,
\cite[Chapter~VII, Section~4]{Feller1968}.

Let \(K_i:=K_i(M)\) and \(K:=(K_1,\dots,K_n)\). Since
\[
\Pr(Y_1=y_1,\dots,Y_M=y_M\mid P)=\prod_{i=1}^n P_i^{K_i},
\]
the sample enters only through the count vector \(K\). Hence \(K\) is sufficient for \(P\), and therefore
\[
I(P;Y_1,\dots,Y_M)=I(P;K).
\]

Let \(h(\cdot)\) denote differential entropy. By Dirichlet--multinomial conjugacy,
\(P\mid K\sim \Dir(\beta_1,\dots,\beta_n)\), where \(\beta_i:=\alpha_i+K_i\) and \(\beta_0:=\sum_i\beta_i=\alpha_0+M\). Thus
\[
I(P;K)=h(P)-\E[h(P\mid K)].
\]
Since \(P\sim \Dir(\alpha_1,\dots,\alpha_n)\), the first term \(h(P)\) is a finite constant independent of \(M\).

For a general Dirichlet law \(\Dir(\beta_1,\dots,\beta_n)\), one has
\[
h(\Dir(\beta))
=
\log B(\beta)+(\beta_0-n)\Psi(\beta_0)-\sum_{i=1}^n (\beta_i-1)\Psi(\beta_i),
\]
where
\[
B(\beta) := B(\beta_1, \ldots, \beta_n) := \frac{\prod_{i=1}^n \Gamma(\beta_i)}{\Gamma(\beta_0)},~~ \beta_0 := \sum_{i=1}^n \beta_i,
\]
denotes the multivariate beta function.
Using Stirling's formula \(\log\Gamma(x)=x\log x - x - \tfrac{1}{2}\log x + O(1)\)
and the digamma asymptotic \(\Psi(x)=\log x+O(x^{-1})\), one computes
\begin{align*}
\log B(\beta) &= \textstyle\sum_i \beta_i\log\beta_i - \beta_0\log\beta_0
                 - \tfrac12\sum_i\log\beta_i + \tfrac12\log\beta_0 + O(1),\\
(\beta_0-n)\Psi(\beta_0) &= \beta_0\log\beta_0 - n\log\beta_0 + O(1),\\
-\textstyle\sum_i(\beta_i-1)\Psi(\beta_i) &=
  -\textstyle\sum_i\beta_i\log\beta_i + \sum_i\log\beta_i + O(1).
\end{align*}
Summing these three lines and simplifying yields, uniformly over all
\(\beta_i\ge\alpha_i>0\),
\[
h(\Dir(\beta))
=
-\frac{n-1}{2}\log \beta_0
+\frac12\sum_{i=1}^n \log\frac{\beta_i}{\beta_0}
+O(1).
\]
Applying this with \(\beta_i=\alpha_i+K_i\) yields
\[
\E[h(P\mid K)]
=
-\frac{n-1}{2}\log(\alpha_0+M)
+\frac12\sum_{i=1}^n \E\!\left[\log\frac{\alpha_i+K_i}{\alpha_0+M}\right]
+O(1).
\]

Set \(Z_{M,i}:=(\alpha_i+K_i)/(\alpha_0+M)\). Since \(P\mid K\sim\Dir(\alpha_1+K_1,\dots,\alpha_n+K_n)\), the posterior mean of \(P_i\) is \(Z_{M,i}=\E[P_i\mid K]\). In particular \(0<Z_{M,i}\le 1\), so \(\log Z_{M,i}\le 0\). On the other hand, by concavity of \(\log\),
\[
\log Z_{M,i}
=
\log \E[P_i\mid K]
\ge
\E[\log P_i\mid K].
\]
Taking expectations gives \(\E[\log Z_{M,i}]\ge \E[\log P_i]\). Since the marginal law of \(P_i\) is \(\Beta(\alpha_i,\alpha_0-\alpha_i)\), we have \(\E[\log P_i]=\Psi(\alpha_i)-\Psi(\alpha_0)>-\infty\). Therefore
\[
\Psi(\alpha_i)-\Psi(\alpha_0)\le \E[\log Z_{M,i}]\le 0,
\]
so \(\E[\log Z_{M,i}]=O(1)\) uniformly in \(M\). It follows that
\[
\E[h(P\mid K)]
=
-\frac{n-1}{2}\log M+O(1).
\]

Finally, since \(h(P)\) is constant,
\[
I(W;X_0^T)
=
I(P;K)
=
h(P)-\E[h(P\mid K)]
=
\frac{n-1}{2}\log M+O(1).
\]
Because \(M=\lceil T/2\rceil\), we have \(\log M=\log T+O(1)\), and hence
\[
I(W;X_0^T)=\frac{n-1}{2}\log T+O(1).
\]

If the walk starts from a leaf rather than from the center, then \(X_1=v_0\) deterministically, so \(I(W;X_0^T)=I(W;X_1^T)\), and the same argument applies after shifting time by one step.
\end{proof}

\section{Proof of Proposition \ref{prop:posterior_conjugacy}}
\label{sec:proof_posterior_conjugacy}

\begin{proof}
Fix a path \(x_0^T\), and write \(N:=N(x_0^T)\) for its undirected edge-count vector and \(x_T:=x_T(x_0^T)\) for its endpoint. Let \(N_v\) denote the number of departures from \(v\) along the path. By the quenched likelihood formula,
\[
\Psrw^{v_0,w}(X_0^T=x_0^T)=\prod_{e\in E} w_e^{N_e}\prod_{v\in V} w_v^{-N_v}.
\]
Multiplying this by the magic-formula density of \(\mu_{v_0,A}\), we obtain
\[
\Psrw^{v_0,w}(X_0^T=x_0^T)\,\mu_{v_0,A}(dw)
\propto
\frac{w_{v_0}^{1/2}\prod_{e\in E} w_e^{a_e-1+N_e}}
{\prod_{v\in V} w_v^{(a_v+1)/2+N_v}}
\sqrt{\tau(w)}\,dw_{-e_0}.
\]

Now set \(a'_e:=a_e+N_e\), so that \(a'_v:=\sum_{e\ni v}a'_e=a_v+\sum_{e\ni v}N_e\). Writing \(d_v:=\sum_{e\ni v}N_e\), we have the standard identity
\[
d_v=2N_v+\1(v=x_T)-\1(v=v_0),
\]
because each departure from \(v\) contributes \(1\) to \(d_v\), each arrival contributes \(1\), and along a path started at \(v_0\) and ending at \(x_T\), arrivals and departures differ by \(\1(v=x_T)-\1(v=v_0)\). Hence
\[
\frac{a_v+1}{2}+N_v
=
\frac{a_v+d_v+1-\1(v=x_T)+\1(v=v_0)}{2}
=
\frac{a'_v+1-\1(v=x_T)+\1(v=v_0)}{2}.
\]
Substituting this into the previous display gives
\[
\Psrw^{v_0,w}(X_0^T=x_0^T)\,\mu_{v_0,A}(dw)
\propto
\frac{w_{v_0}^{1/2}\prod_{e\in E} w_e^{a'_e-1}}
{\prod_{v\in V} w_v^{(a'_v+1-\1(v=x_T)+\1(v=v_0))/2}}
\sqrt{\tau(w)}\,dw_{-e_0}.
\]

The extra vertex factor simplifies because
\[
\frac{w_{v_0}^{1/2}}
{\prod_{v\in V} w_v^{(-\1(v=x_T)+\1(v=v_0))/2}}
=
w_{x_T}^{1/2}.
\]
Therefore
\[
\Psrw^{v_0,w}(X_0^T=x_0^T)\,\mu_{v_0,A}(dw)
\propto
\frac{w_{x_T}^{1/2}\prod_{e\in E} w_e^{a'_e-1}}
{\prod_{v\in V} w_v^{(a'_v+1)/2}}
\sqrt{\tau(w)}\,dw_{-e_0}.
\]
This is exactly the magic-formula density of \(\mu_{x_T,A+N}\), up to normalization. Restoring the normalizing constants yields
\[
\Psrw^{v_0,w}(X_0^T=x_0^T)\,\mu_{v_0,A}(dw)
=
\frac{Z_{x_T,A+N}}{Z_{v_0,A}}\,
\mu_{x_T,A+N}(dw),
\]
which is \eqref{eq:unnormalized_bayes_magic_main}.

Integrating over \(w\) gives
\[
\Perrw^{v_0,A}(X_0^T=x_0^T)=\frac{Z_{x_T,A+N}}{Z_{v_0,A}},
\]
which is \eqref{eq:path_probability_Z_ratio_main}. Dividing the previous identity by this path probability yields
\[
\mu_{v_0,A}(dw\mid X_0^T=x_0^T)=\mu_{x_T,A+N}(dw),
\]
which is \eqref{eq:posterior_path_main}.

Finally, for a path started at \(v_0\), the endpoint is determined by the parity rule
\begin{equation}    \label{eq:parityrule}
\sum_{e\ni v}N_e \equiv \1(v=v_0)+\1(v=x_T)\pmod 2.
\end{equation}
Thus \(x_T\) is a function of \(N\), and the posterior depends on the realized path only through \(N\). Hence
\[
\mu_{v_0,A}(dw\mid N)=\mu_{x_T(N),A+N}(dw),
\]
which is \eqref{eq:posterior_count_main}.
\end{proof}

\section{Proof of \Cref{thm:nstar_inverse_moment}}
\label{sec:proof_nstar_inverse}

We retain the notation from 
\cref{sec:n-star-asyptotics}:
\(G\) is the \(n\)-star with center \(v_0\), all initial edge weights are equal to \(a>0\), \(e\) is a fixed leaf edge, and \(N_e(T)\),
with an abuse of notation,
denotes the number of the first \(T\) excursions that use \(e\).

The first ingredient is the standard beta-binomial representation.

\begin{lemma}[Beta-binomial mixture]
\label{lem:beta_binomial_star}
Let \(b:=(n-1)a\). There exists a random variable \(X\sim \Beta(a,b)\) such that, conditional on \(X=x\), one has \(N_e(T)\mid X=x\sim \Bin(T,x)\). Consequently, for every \(\gamma>0\),
\begin{equation}
\label{eq:star_mixture}
\E\!\left[\frac{1}{N_e(T)+\gamma}\right]
=
\frac{1}{B(a,b)}
\int_0^1 x^{a-1}(1-x)^{b-1}\,g_T(x)\,dx,
\end{equation}
where \(g_T(x):=\E[(\Bin(T,x)+\gamma)^{-1}]\).
\end{lemma}

\begin{proof}
The sequence of leaf edges chosen from the center is exactly a symmetric P\'olya urn with \(n\) colors and initial weight \(a\) in each color. Hence the empirical frequencies are mixed by a \(\Dir(a,\dots,a)\) law, and the marginal frequency of one fixed edge is \(\Beta(a,b)\) with \(b=(n-1)a\). Conditional on this frequency \(X=x\), the number of uses of \(e\) among the first \(T\) excursions is \(\Bin(T,x)\). Integrating out \(X\) gives \eqref{eq:star_mixture}.
\end{proof}

For the proof of the theorem, we only need upper bounds on the binomial inverse moment.

\begin{lemma}[Upper bounds for the binomial inverse moment]
\label{lem:binomial_inverse_upper}
Let \(c_\gamma:=\max\{1,\gamma^{-1}\}\). Then, for every \(T\ge 2\),
\[
g_T(x)\le c_\gamma \qquad \text{for } 0<x\le \frac{1}{2T},
\]
and
\[
g_T(x)\le \frac{6c_\gamma}{Tx} \qquad \text{for } \frac{1}{2T}\le x\le 1.
\]
\end{lemma}

\begin{proof}
Let \(K\sim \Bin(T,x)\). 
We have \((K+\gamma)^{-1}\le c_\gamma\) pointwise, so \(g_T(x):=\E[(K+\gamma)^{-1}] \le c_\gamma\)
for all $x \in [0,1]$ and so, in particular, for \(0<x\le (2T)^{-1}\).

Now assume \(x\ge (2T)^{-1}\), so \(Tx\ge 1/2\). Write
\[
g_T(x)=\E\!\left[\frac{1}{K+\gamma}\mathbf 1_{\{K\le Tx/2\}}\right]
      +\E\!\left[\frac{1}{K+\gamma}\mathbf 1_{\{K>Tx/2\}}\right].
\]
On \(\{K>Tx/2\}\), we have \((K+\gamma)^{-1}\le 2/(Tx)\). On the other hand, \((K+\gamma)^{-1}\le \gamma^{-1}\le c_\gamma\), so
\[
g_T(x)\le c_\gamma\,\P(K\le Tx/2)+\frac{2}{Tx}.
\]
Since \(\E K=Tx\), the standard Chernoff--Hoeffding lower-tail bound
\cite{Hoeffding1963} gives
\[
\P(K\le Tx/2)\le e^{-Tx/8}.
\]
Also, for all \(y>0\), one has \(e^{-y/8}\le 4/y\) because \(\sup_{y>0} y e^{-y/8}=8/e<4\). Applying this with \(y=Tx\) yields \(e^{-Tx/8}\le 4/(Tx)\), and therefore
\[
g_T(x)\le \frac{4c_\gamma}{Tx}+\frac{2}{Tx}\le \frac{6c_\gamma}{Tx}.
\]
\end{proof}

We now prove the theorem with explicit constants. Write \(b:=(n-1)a\), set \(M_b:=2^{(1-b)_+}\), and define
\[
C_{a,n,\gamma}^{(>)}:=
\frac{c_\gamma M_b}{a\,2^a\,B(a,b)}
+\frac{6c_\gamma M_b\,2^{1-a}}{(a-1)B(a,b)}
+\frac{6c_\gamma\,2^{(2-a)_+}}{b\,B(a,b)},
\]
for \(a>1\),
\[
C_{n,\gamma}^{(=)}:=
\frac{c_\gamma M_b}{2\,B(1,n-1)}
+\frac{6c_\gamma M_b}{B(1,n-1)}
+\frac{12c_\gamma}{(n-1)B(1,n-1)},
\]
for \(a=1\), and
\[
C_{a,n,\gamma}^{(<)}:=
\frac{c_\gamma M_b}{a\,2^a\,B(a,b)}
+\frac{6c_\gamma M_b\,2^{1-a}}{(1-a)B(a,b)}
+\frac{6c_\gamma\,2^{(2-a)_+}}{b\,B(a,b)},
\]
for \(0<a<1\).

\begin{proof}[Proof of \Cref{thm:nstar_inverse_moment}]
Let
\[
M_T:=\E\!\left[\frac{1}{N_e(T)+\gamma}\right].
\]
By Lemma \ref{lem:beta_binomial_star}, we have
\[
M_T
=
\frac{1}{B(a,b)}
\left(
\int_0^{1/(2T)} x^{a-1}(1-x)^{b-1} g_T(x)\,dx
+
\int_{1/(2T)}^1 x^{a-1}(1-x)^{b-1} g_T(x)\,dx
\right).
\]
Using Lemma \ref{lem:binomial_inverse_upper}, we obtain
\[
M_T
\le
\frac{c_\gamma}{B(a,b)}
\int_0^{1/(2T)} x^{a-1}(1-x)^{b-1}\,dx
+
\frac{6c_\gamma}{T\,B(a,b)}
\int_{1/(2T)}^1 x^{a-2}(1-x)^{b-1}\,dx.
\]

We estimate the two integrals separately. On \([0,1/2]\), we have \((1-x)^{b-1}\le M_b\), where \(M_b=2^{(1-b)_+}\). Hence
\[
\int_0^{1/(2T)} x^{a-1}(1-x)^{b-1}\,dx
\le
M_b\int_0^{1/(2T)}x^{a-1}\,dx
=
\frac{M_b}{a\,2^a}\,T^{-a}.
\]
For the second integral, split at \(x=1/2\). On \([1/(2T),1/2]\), we still have \((1-x)^{b-1}\le M_b\), so
\[
\int_{1/(2T)}^{1/2} x^{a-2}(1-x)^{b-1}\,dx
\le
M_b\int_{1/(2T)}^{1/2}x^{a-2}\,dx.
\]
On \([1/2,1]\), we have \(x^{a-2}\le 2^{(2-a)_+}\), and therefore
\[
\int_{1/2}^1 x^{a-2}(1-x)^{b-1}\,dx
\le
2^{(2-a)_+}\int_{1/2}^1 (1-x)^{b-1}\,dx
\le
\frac{2^{(2-a)_+}}{b}.
\]

We now distinguish the three regimes.

If \(a>1\), then
\[
\int_{1/(2T)}^{1/2}x^{a-2}\,dx
\le
\int_0^{1/2}x^{a-2}\,dx
=
\frac{2^{1-a}}{a-1}.
\]
Thus
\[
M_T
\le
\frac{c_\gamma M_b}{a\,2^a\,B(a,b)}\,T^{-a}
+
\frac{6c_\gamma}{T\,B(a,b)}
\left(
\frac{M_b\,2^{1-a}}{a-1}
+
\frac{2^{(2-a)_+}}{b}
\right).
\]
Since \(a>1\) and \(T\ge 1\), one has \(T^{-a}\le T^{-1}\). Therefore
\[
M_T\le C_{a,n,\gamma}^{(>)}\,T^{-1}.
\]

If \(a=1\), then \(b=n-1\) and
\[
\int_{1/(2T)}^{1/2}x^{-1}\,dx=\log T.
\]
Hence
\[
M_T
\le
\frac{c_\gamma M_b}{2\,B(1,n-1)}\,T^{-1}
+
\frac{6c_\gamma}{T\,B(1,n-1)}
\left(
M_b\log T+\frac{2}{n-1}
\right).
\]
Since \(\log T\le 1+\log T\) and \(1\le 1+\log T\), this implies
\[
M_T\le C_{n,\gamma}^{(=)}\,\frac{1+\log T}{T}.
\]

If \(0<a<1\), then
\[
\int_{1/(2T)}^{1/2}x^{a-2}\,dx
=
\frac{(1/(2T))^{a-1}-(1/2)^{a-1}}{1-a}
\le
\frac{2^{1-a}}{1-a}\,T^{1-a}.
\]
Thus
\[
M_T
\le
\frac{c_\gamma M_b}{a\,2^a\,B(a,b)}\,T^{-a}
+
\frac{6c_\gamma}{T\,B(a,b)}
\left(
\frac{M_b\,2^{1-a}}{1-a}\,T^{1-a}
+
\frac{2^{(2-a)_+}}{b}
\right).
\]
Since \(0<a<1\) and \(T\ge 1\), one has \(T^{-1}\le T^{-a}\). Therefore
\[
M_T\le C_{a,n,\gamma}^{(<)}\,T^{-a}.
\]

This proves the theorem.
\end{proof}

\section{Proof of Lemma~\ref{lem:A-lower-bound-explicit}}
\label{sec:proof_A_lower_bound}

\begin{proof}
Fix a nonempty proper subset $F\subsetneq E$, and write
\[
S:=S(F).
\]
Partition the edges of $F$ into three classes:
\begin{align*}
F_2&:=\{g\in F:\ \text{both endpoints of }g\text{ belong to }S\},\\
F_1&:=\{g\in F:\ \text{exactly one endpoint of }g\text{ belongs to }S\},\\
F_0&:=\{g\in F:\ \text{no endpoint of }g\text{ belongs to }S\}.
\end{align*}
Then
\[
F=F_0\sqcup F_1\sqcup F_2.
\]

We have
\[
\sum_{v\in S} a_v = 2a(F_2)+a(F_1).
\]
Therefore
\[
a(F)-\frac12\sum_{v\in S} a_v
=
a(F_0)+\frac12 a(F_1).
\tag{1}
\]

Next,
\[
\sum_{v\in S} b_v
=
\frac12\sum_{v\in S} a_v+\frac12 |S\setminus\{v_0\}|,
\]
so
\[
A(F)
=
\Bigl(a(F)-\frac12\sum_{v\in S} a_v\Bigr)
+
\frac{\kappa(G\setminus F)-1-|S\setminus\{v_0\}|}{2}.
\tag{2}
\]

We now estimate the two terms on the right-hand side.

\smallskip

\noindent
\emph{First term.}
By \((1)\),
\[
a(F)-\frac12\sum_{v\in S} a_v
=
a(F_0)+\frac12 a(F_1).
\]
Since $a_g\ge \underline a$ for every edge,
\[
a(F_0)+\frac12 a(F_1)
\ge
\underline a\,|F_0|+\frac{\underline a}{2}|F_1|.
\tag{3}
\]

We claim that $F_0\cup F_1\neq\varnothing$. Indeed, suppose by contradiction that $F_0=F_1=\varnothing$. Then every edge in $F$ belongs to $F_2$, so both endpoints of every edge in $F$ belong to $S$. Since every vertex in $S$ has all its incident edges in $F$, there is no edge in $E\setminus F$ incident to any vertex of $S$. Hence every connected component of the subgraph $(S,F)$ is also a connected component of $G$. Because $F$ is nonempty, $(S,F)$ has at least one edge, so $G$ has a connected component contained in $(S,F)$. Since $F\subsetneq E$, there is also at least one edge in $E\setminus F$, hence at least one vertex not in $S$. This contradicts the connectedness of $G$. Therefore $F_0\cup F_1\neq\varnothing$, and thus \((3)\) implies
\[
a(F)-\frac12\sum_{v\in S} a_v \ge \frac{\underline a}{2}.
\tag{4}
\]

\smallskip

\noindent
\emph{Second term.}
Every vertex in $S$ becomes isolated in $G\setminus F$, because all incident edges of such a vertex belong to $F$ and are removed. Hence $G\setminus F$ has at least $|S|$ connected components coming from these isolated vertices.

Since $F\subsetneq E$, there exists an edge in $E\setminus F$. The connected component of $G\setminus F$ containing that edge is different from all the isolated-vertex components coming from $S$. Therefore
\[
\kappa(G\setminus F)\ge |S|+1.
\] 
Consequently,
\[
\kappa(G\setminus F)-1-|S\setminus\{v_0\}|
\ge
|S|-|S\setminus\{v_0\}|
\ge 0.
\tag{5}
\]

Combining \((2)\), \((4)\), and \((5)\), we obtain
\[
A(F)\ge \frac{\underline a}{2}.
\]
This proves the lemma.
\end{proof}

\end{document}